\documentclass[
twocolumn,
superscriptaddress,
amsmath, amssymb, aps,
prb,
longbibliography
]{revtex4-2}

\usepackage[dvipdfmx]{graphicx}
\usepackage{dcolumn}
\usepackage{bm}
\usepackage{hyperref}
\usepackage{xcolor}

\hypersetup{
  colorlinks=true,
  linkcolor=[rgb]{0.60,0.00,0.00},
  citecolor=[rgb]{0.00,0.00,0.60},
  urlcolor=[rgb]{0.00,0.00,0.60},
  setpagesize=false
}

\usepackage{amsthm}
\newtheorem{thm}{Theorem}
\newtheorem{lemma}{Lemma}

\usepackage{algpseudocode}
\usepackage{algorithm}

\begin{document}
\title{Improved real-space parallelizable matrix-product state compression and \\ its application to unitary quantum dynamics simulation}

\begin{abstract}
Towards the efficient simulation of near-term quantum devices using tensor network states, 
we introduce an improved real-space parallelizable matrix-product state (MPS) compression method. 
This method enables efficient compression of all virtual bonds in constant time, irrespective of the system size, with controlled 
accuracy, while it maintains the stability of the wavefunction norm without necessitating sequential renormalization procedures. 
In addition, we introduce a parallel regauging technique to partially restore the deviated canonical form, thereby improving 
the accuracy of the simulation in subsequent steps. We further apply this method to simulate unitary quantum dynamics and introduce an improved parallel time-evolving block-decimation 
(pTEBD) algorithm. We employ the improved pTEBD algorithm for extensive simulations of typical one- and two-dimensional quantum 
circuits, involving over 1000 qubits. The obtained numerical results unequivocally demonstrate that the improved pTEBD algorithm  
achieves the same level of simulation precision as the current state-of-the-art MPS algorithm but in polynomially shorter 
time, exhibiting nearly perfect weak scaling performance on a modern supercomputer.
\end{abstract}

\author{Rong-Yang Sun}
\affiliation{Computational Materials Science Research Team, RIKEN Center for Computational Science (R-CCS), Kobe, Hyogo, 650-0047, Japan}
\affiliation{Quantum Computational Science Research Team, RIKEN Center for Quantum Computing (RQC), Wako, Saitama, 351-0198, Japan}
\affiliation{RIKEN Interdisciplinary Theoretical and Mathematical Sciences Program (iTHEMS), Wako, Saitama 351-0198, Japan}
\author{Tomonori Shirakawa}
\affiliation{Computational Materials Science Research Team, RIKEN Center for Computational Science (R-CCS), Kobe, Hyogo, 650-0047, Japan}
\affiliation{Quantum Computational Science Research Team, RIKEN Center for Quantum Computing (RQC), Wako, Saitama, 351-0198, Japan}
\affiliation{RIKEN Interdisciplinary Theoretical and Mathematical Sciences Program (iTHEMS), Wako, Saitama 351-0198, Japan}
\affiliation{Computational Condensed Matter Physics Laboratory, RIKEN Cluster for Pioneering Research (CPR), Saitama 351-0198, Japan}
\author{Seiji Yunoki}
\affiliation{Computational Materials Science Research Team, RIKEN Center for Computational Science (R-CCS), Kobe, Hyogo, 650-0047, Japan}
\affiliation{Quantum Computational Science Research Team, RIKEN Center for Quantum Computing (RQC), Wako, Saitama, 351-0198, Japan}
\affiliation{Computational Condensed Matter Physics Laboratory, RIKEN Cluster for Pioneering Research (CPR), Saitama 351-0198, Japan}
\affiliation{Computational Quantum Matter Research Team, RIKEN Center for Emergent Matter Science (CEMS), Wako, Saitama 351-0198, Japan}

\date{\today}

\maketitle

\section{Introduction}
Owing to the recent rapid advancement of synthetic quantum 
devices~\cite{gross2017quantum,arute2019quantum,schafer2020tools,pogorelov2021compact}, 
quantum computing is becoming a competitive candidate for next-generation computing. 
It utilizes these controllable quantum systems for quantum 
information processing and to address challenging problems beyond the capability of classical computing~\cite{nielsen2010quantum}.  
Due to the principles of quantum mechanics, a quantum computer maintaining a quantum state composed of $N$ qubits 
can access a total of $2^{N}$ dimensions in its working space, i.e., the Hilbert space,  
hence providing ample computational resources for addressing problems 
with comparable exponential complexity, 
such as the quantum many-body problem~\cite{daley2022practical}.
An ideal quantum computer can encompass the entire Hilbert space, thereby offering indisputable exponential resources 
compared to a classical computer, which thus establishes the quantum advantage. 
However, currently available quantum computers are hindered by strong decoherence noise, allowing only a small fraction of 
the Hilbert space to be explored~\cite{zhou2020what}. 
These noisy quantum computers are referred to as a noisy intermediate-scale quantum (NISQ)~\cite{preskill2018quantum} device.
The central challenge in near-term development of quantum computing is determining efficient strategies for utilizing NISQ devices 
to achieve a practical quantum advantage in meaningful computational tasks~\cite{daley2022practical,kim2023evidence,Alexeev2023,Moreno2024}.

In parallel, to validate the attainment of practical quantum advantage and to establish a reliable and readily accessible testing 
environment for quantum algorithm development, extensive research has been focused on developing classical simulations of 
quantum computing. 
For exact simulations, two primary methods are employed: the statevector method and the tensor contraction method. 
The statevector method can simulate quantum circuits with any circuit depth (in polynomial time with respect to the circuit depth) 
and provides the complete amplitudes of 
the wavefunction. However, simulating a quantum circuit with more than 50 qubits using currently 
available classical computers is considered impossible~\cite{haner2017petabyte,deraedt2019massively}. 
On the other hand, the tensor contraction method~\cite{guo2019general,gray2021hyper,guo2021verifying,pan2022simulation,liu2024verifying} can simulate 
quantum circuits with up to around 100 qubits,
but the arrangement of gates and the circuit depth are severely restricted, allowing only a limited portion of the wavefunction 
amplitudes to be obtained~\cite{liu2021closing}. 
Therefore, to efficiently simulate current NISQ devices with over 400 qubits~\cite{collins2022ibm}, it is essential to develop 
more specialized simulation algorithms. These algorithms should consider the limitation that NISQ devices can only 
generate limited quantum entanglement.

Inspired by their great successes in the study of quantum many-body problems~\cite{orus2014a,cirac2021matrix}, 
tensor network states, particularly the matrix-product state (MPS) \cite{schollwock2011the} and the projected entangled-pair state (PEPS) \cite{verstraete2008matrix}, 
have been employed to simulate quantum computing and have achieved accurate results in faithfully simulating NISQ 
devices~\cite{zhou2020what,ayral2023density,tindall2023efficient}. 
In this paper, our focus will be on MPS-based simulation algorithms.
Presently, MPS-based methods can perform full-amplitude approximate simulations for quantum circuits with over 100 qubits 
and moderate circuit depth, benefiting from its efficient representation of quantum entanglement~\cite{verstraete2006matrix}. 
Nevertheless, restricted by the real-space sequential nature of these algorithms~\cite{zhou2020what,ayral2023density}, efficiently simulating a NISQ device with hundreds, or even thousands, of qubits remains elusive.

In general, an MPS-based algorithm for quantum circuit simulations includes two main procedures. One involves applying quantum 
gates, and the other involves compressing the MPS to a computable size. Regarding the former, we can parallelly apply 
multiple gates as long as they do not have spacetime overlap. 
For the later, it has been realized that even if all the MPS virtual bonds 
are compressed in parallel, the error induced by this parallel compression is still manageable~\cite{verstraete2006matrix}. 
These observations imply the feasibility of developing a real-space parallelizable algorithm to simulate quantum circuits. 
In fact, in the study of the time evolution of a quantum many-body system, a similar parallel scheme has been proposed 
and has obtained promising results~\cite{urbanek2016parallel} (see also Ref.~\cite{secular2020parallel} 
for an alternative parallel approach to simulating Hamiltonian dynamics and Ref.~\cite{stoudenmire2013real} for a parallel 
algorithm of searching the ground state).
However, a crucial issue that remains in this scheme is the fast-decaying wavefunction norm caused by parallel MPS compression, 
which might result in serious numerical instability. This instability necessitates the renormalization of the the wavefunction 
after simulating a certain period of time evolution~\cite{urbanek2016parallel}. 
Being a sequential procedure, this additional wavefunction renormalization significantly diminishes parallelism and hampers 
the efficient utilization of the parallel computing environment.

Here, we introduce an improved real-space parallelizable MPS compression (IPMC) method that can stabilize 
the wavefunction norm without compromising parallelism. 
For an MPS in the canonical form, we prove that the wavefunction norm, after undergoing the IPMC, stabilizes within an interval 
with two bounds uniformly converging to 1 from both sides, and the convergence improves as the compression error decreases. 
Moreover, we integrate an additional parallel regauging procedure into the IPMC to generate a better starting point for 
the next simulation step, thereby enhancing the simulation accuracy. 
Based on this IPMC method, we propose the parallel time-evolving block-decimation (pTEBD) algorithm to simulate unitary quantum 
dynamics in a fully parallelizable manner. 
We benchmark the pTEBD algorithm by extensively simulating typical random and parametrized quantum circuits on both 
one-dimensional (1D) and two-dimensional (2D) qubit arrays. 
These numerical results demonstrate that the pTEBD algorithm achieves a comparable simulation precision to the previous 
sequential MPS algorithm proposed in \cite{zhou2020what} and can attain the same accuracy in polynominally shorter time. 
Meanwhile, the wavefunction norm stabilizes to approach 1, instead of decaying exponentially to zero, 
throughout the entire simulation period with a circuit depth exceeding 1000. Consequently, we successfully achieve nearly perfect 
weak scaling performance in the pTEBD simulation, involving over 1000 qubits with over 250 computational nodes on the Supercomputer Fugaku installed at RIKEN.

The rest of the paper is organized as follows.
First, we briefly summarize the MPS representation of quantum many-body states, the properties of the canonical form, and 
the time-evolving block-decimation (TEBD) algorithm (also denoted as the simple update algorithm) in Sec.~\ref{sec:mps}.
Next, we systematically describe the real-space parallelizable MPS compression method in Sec.~\ref{sec:mps_comp}, 
where we introduce 
a wavefunction norm stabilization method with a proof of its bounding theorem, explore the parallel regauging technique 
through the trivial simple update, and finally explain the IPMC method. 
In Sec.~\ref{sec:app_uqd}, we propose the pTEBD algorithm and demonstrate its accuracy, numerical stability, and performance 
through extensive simulations of random and parametrized quantum circuits in one and two dimensions. 
Finally, in Sec.~\ref{sec:sum}, we summarize the results of this paper and briefly discuss their impact on the development 
of near-term quantum computing. 
The details of the quantum circuits used for the simulations are described in Appendix~\ref{app:qc_setup}. 
An additional result regarding the accuracy of the pTEBD algorithm is provided in Appendix~\ref{app:Fs_fixed_seed}. 
Furthermore, the pTEBD algorithm is applied to the simulation of the quantum Fourier transformation in Appendix~\ref{app:QFT_pTEBD}.

\section{\label{sec:mps} MPS representation of quantum many-body states}

\subsection{Matrix-product state}

The Hilbert space $\mathcal{H}$ for an $N$-site quantum many-body system composed of qubits is described by the tensor product 
of $N$ local Hilbert spaces $\mathcal{H}_{i}$ (spanned by orthogonal states $|0\rangle_{i}$ and $|1\rangle_{i}$ with the dimension 
$d = 2$) located on each qubit site $i$, i.e., 
\begin{equation}
    \mathcal{H} = \mathcal{H}_{1} \otimes \mathcal{H}_{2} \otimes \cdots \otimes \mathcal{H}_{N}~.
\end{equation}
Therefore, the dimension of $\mathcal{H}$ is $\mathrm{dim}(\mathcal{H}) = 2^{N}$, increasing exponentially with $N$. 
Any quantum many-body state $|\Psi\rangle$ living on $\mathcal{H}$ can be represented using a complete many-body basis 
formed by a direct product of local states $|\sigma_{i}\rangle_{i}$ with $\sigma_{i} = 0$ or $1$ on each $\mathcal{H}_{i}$, i.e., 
$|\sigma_{1}\cdots \sigma_{N}\rangle = |\sigma_{1}\rangle_{1}\otimes|\sigma_{2}\rangle_{2}\otimes\cdots\otimes|\sigma_{N}\rangle_{N}$. With this basis, the quantum state $|\Psi\rangle$ is represented as
\begin{equation}
    |\Psi\rangle = \sum_{\sigma_{1},\cdots,\sigma_{N}} \Psi_{\sigma_{1},\cdots,\sigma_{N}} |\sigma_{1}\cdots \sigma_{N}\rangle~,
\end{equation}
where the coefficients $\Psi_{\sigma_{1},\cdots,\sigma_{N}}$ can be regarded as a rank-$N$ tensor. Note that, without specifying, 
the normalization condition $\langle\Psi|\Psi\rangle = 1$ is always maintained.

The MPS representation of $|\Psi\rangle$ is constructed by a tensor network decomposition, i.e., decomposing a tensor 
to a set of tensors (a tensor network) such that the contraction of these tensors restores the original tensor, as 
\begin{equation}
    \label{eq:mps_decomp}
    \Psi_{\sigma_{1},\cdots,\sigma_{N}} = \mathrm{Tr}(M^{[1]\sigma_{1}}M^{[2]\sigma_{2}} \cdots M^{[N]\sigma_{N}})~,
\end{equation}
where $M^{[i]\sigma_{i}}$ is a rank-$3$ tensor on site $i$ and it is simply a matrix when the index $\sigma_{i}$ is fixed. 
The row and column of the matrix $M^{[i]\sigma_{i}}$ are associated with the virtual spaces $\mathcal{V}_{i-1}$ and 
$\mathcal{V}_{i}$, respectively, emerged from the decomposition. The dimension of $\mathcal{V}_{i}$, 
$\chi_{i} = \mathrm{dim}(\mathcal{V}_{i})$, is the bond dimension of bond $i$ connecting sites $i$ and $i+1$. 
Especially, bond $N$ (also labeled as bond $0$) connects sites $N$ and $1$. 
Without losing generality, we can always assign $\chi_{0} = \chi_{N} = 1$ in the decomposition~\cite{schollwock2011the}. 
The corresponding MPS is usually named an \textit{open} MPS. In the rest of the paper, we will only consider the open MPS.

\subsection{Canonical form}

Although the decomposition in Eq.~(\ref{eq:mps_decomp}) can be carried out in various ways, the canonical 
decomposition~\cite{vidal2003efficient}, known for its additional advantageous properties, is commonly employed. 
The canonical decomposition is represented as
\begin{equation}
    \label{eq:mps_cano_decomp}
    \Psi_{\sigma_{1},\cdots,\sigma_{N}} = \mathrm{Tr}(\Lambda^{[0]} \Gamma^{[1]\sigma_{1}} \Lambda^{[1]} \Gamma^{[2]\sigma_{2}} \cdots \Lambda^{[N-1]} \Gamma^{[N]\sigma_{N}} \Lambda^{[N]})~,
\end{equation}
where an additional diagonal real matrix $\Lambda^{[i]}$ is attached to bond $i$ with 
$\Lambda^{[0]} = \Lambda^{[N]} \equiv (1)$ for the open MPS,  
and $\Gamma^{[i]\sigma_{i}}$ is a rank-$3$ tensor on site $i$. 
Here, we assume that the diagonal elements of each $\Lambda^{[i]}$ are in descending order. 
The form $\{\Lambda^{[i]}, \Gamma^{[j]}\}$ in 
Eq.~(\ref{eq:mps_cano_decomp}) is commonly known as the Vidal form. 
Additionally, the MPS is considered to be in the canonical form when $\Gamma^{[i]}$ and $\Lambda^{[i]}$ satisfy 
the canonical conditions:
\begin{eqnarray}
    \label{eq:l_r_cano}
    \nonumber
    A^{[i]\sigma_{i}} &\equiv& \Lambda^{[i-1]} \Gamma^{[i]\sigma_{i}}~\mathrm{with}~\sum_{\sigma_{i}}\bar{A}^{[i]\sigma_{i}}A^{[i]\sigma_{i}} = I_{\mathcal{V}_{i}}\\\nonumber
    B^{[i]\sigma_{i}} &\equiv& \Gamma^{[i]\sigma_{i}} \Lambda^{[i]}~\mathrm{with}~\sum_{\sigma_{i}}B^{[i]\sigma_{i}}\bar{B}^{[i]\sigma_{i}} = I_{\mathcal{V}_{i-1}}~,\\
\end{eqnarray}
where $\bar{A}^{[i]}$ ($\bar{B}^{[i]}$) denotes the complex conjugate transpose of $A^{[i]}$ ($B^{[i]}$) and $I_{\mathcal{V}_{i}}$ 
is the identity matrix on $\mathcal{V}_{i}$.

Here, we concisely list main properties of the canonical form. 
First, we can straightforwardly show that under the complete orthogonal basis
\begin{equation}
    \left\{ |L_{\alpha_i}\rangle = \sum_{\sigma_{1},\cdots,\sigma_{i}} (\Lambda^{[0]} \Gamma^{[1]\sigma_{1}} \cdots \Lambda^{[i-1]} \Gamma^{[i]\sigma_{i}})_{\alpha_i} |\sigma_{1}\cdots \sigma_{i}\rangle\right\}~,
\end{equation}
the reduced density matrix $\rho^{[1\cdots i]}$ of the subsystem $\{1,2,\cdots, i\}$ is simply represented as
\begin{equation}
\rho^{[1\cdots i]} = (\Lambda^{[i]})^{2}~\text{with}~\sum_{\alpha_{i}}(\Lambda^{[i]}_{\alpha_{i}\alpha_{i}})^{2} = 1~.
\end{equation}
Note that the elements of the diagonal matrix $\Lambda^{[i]}$ correspond to the singular value spectrum $\{\lambda_{\alpha_{i}}\}$ 
obtained through the singular value decomposition (SVD) of the matrix 
$\Psi_{\sigma_{1}\cdots\sigma_{i}, \sigma_{i+1}\cdots\sigma_{N}}$, where  
the rows and columns are formed by grouping indices $\{\sigma_{1},\cdots,\sigma_{i}\}$ and 
$\{\sigma_{i+1},\cdots,\sigma_{N}\}$, respectively. 
Therefore, $(\Lambda^{[i]})^{2}$ is also the reduced density matrix of the subsystem $\{i+1,i+2,\cdots, N\}$.

Using these properties, the information of quantum entanglement for $|\Psi\rangle$ can be directly extracted from the canonical form. Naturally, the bond $i$ provides a bipartition of the system into two subsystems $\{1,\cdots,i\}$ and $\{i+1,\cdots, N\}$. 
Accordingly, the entanglement entropy $S_{v}^{[i]}$ for this bipartition is given by 
\begin{equation}
    S_{v}^{[i]} = - \mathrm{Tr}(\rho^{[1\cdots i]} \ln \rho^{[1\cdots i]}) = -\sum_{\alpha_{i}} \lambda_{\alpha_{i}}^{2} \ln \lambda_{\alpha_{i}}^{2}~,
\end{equation}
and the entanglement spectrum level $\xi_{\alpha_i}$ is given by $\xi_{\alpha_i} = -\ln \lambda_{\alpha_{i}}^{2}$. 
Therefore, truncating the bond dimension $\chi_i$ at bond $i$ on the basis of $\{\lambda_{\alpha_{i}}\}$ allows for the retention 
of a significant portion of the entanglement information between these two subsystems.
For practical construction of the canonical form for a general quantum state and additional insights into its useful properties, 
one can find more comprehensive reviews~\cite{schollwock2011the,orus2014a,cirac2021matrix}.

\subsection{TEBD algorithm}

Based on the canonical form, the unitary dynamics of a quantum state can be efficiently simulated using the TEBD 
algorithm~\cite{vidal2003efficient}. In general, a unitary quantum dynamic process is described by applying a unitary 
operator $\hat{U}$ to the state $|\Psi\rangle$. However, $\hat{U}$ can be decomposed into 
many two-site local unitary operators $\hat{U}^{\alpha}_{i, i+1}$ acting on sites $i$ and $i+1$, where $\alpha$ is 
the index for these local operators~\cite{shirakawa2021automatic}. Therefore, we can focus on simulating 
$|\Psi\rangle \rightarrow \hat{U}^{\alpha}_{i, i+1}|\Psi\rangle$ at a specific moment. The simulation of this process 
using the TEBD algorithm is referred to as a TEBD procedure and will be outlined below.

A TEBD procedure consists of three steps. Starting from the canonical form of $|\Psi\rangle$, we contract the involved local tensors 
to construct the local two-site wavefunction
\begin{equation}
    \Psi_{i, i+1} = \Lambda^{[i-1]}\Gamma^{[i]}\Lambda^{[i]}\Gamma^{[i+1]}\Lambda^{[i+1]}~.
\end{equation}
Next, we update $\Psi_{i, i+1}$ to $\Psi_{i, i+1}'$ by appending $\hat{U}^{\alpha}_{i, i+1}$ to $\Psi_{i, i+1}$. 
Subsequently, we perform the SVD of $\Psi_{i, i+1}' = U^{[i]} S \bar{V}^{[i+1]}$ and update the canonical form as
\begin{eqnarray}
    \nonumber
    \Gamma^{[i]} &\leftarrow& (\Lambda^{[i-1]})^{-1} U^{[i]}~,\\
    \nonumber
    \Lambda^{[i]} &\leftarrow& S~,\\
    \Gamma^{[i+1]} &\leftarrow& \bar{V}^{[i+1]} (\Lambda^{[i+1]})^{-1}
\end{eqnarray}
with $\Lambda^{[i-1]}$ and $\Lambda^{[i+1]}$ intact.

Ideally, the TEBD procedure can be parallelized perfectly in the real-space to evaluate multiple local unitary operators 
$\{\hat{U}^{\alpha}_{i, i+1}\}$ simultaneously as long as these unitary operators do not overlap with each other, automatically 
maintaining the canonical form. However, the bond dimensions in general increase exponentially. 
The TEBD algorithm has also been applied to simulate imaginary-time evolution using 
an MPS~\cite{vidal2004efficient,orus2008infinite} and a PEPS~\cite{jiang2008accurate} for ground state searches, and   
has also become a fundamental technique in the research field of tensor network states under the name 
of the simple update (SU) method. Specifically, in the case where $\hat{U}^{\alpha}_{i, i+1}$ is an identity operator, the TEBD or SU procedure does not 
alter the underlying quantum state $|\Psi\rangle$ and is thus denoted as the trivial simple update (tSU) step. 
The repeated application of tSU steps to an MPS or a PEPS, known as the tSU algorithm, 
has been developed to construct its canonical form 
for an MPS or its quasicanonical form for a PEPS~\cite{kalis2012fate,ran2012optimized,alkabetz2021tensor}.

\subsection{Computational complexities of MPS calculations}

Before concluding this section, let us discuss the spacial and time complexities of generic MPS calculations. 
For an MPS with $N$ sites and a bond dimension $\chi$, the spacial complexity, i.e., the memory footprint, is 
$\mathcal{O}(N\chi^{2})$, increasing polynomially with both $N$ and $\chi$. Here, for simplicity, we assume the same bond 
dimension for each matrix and maintain this convention throughout the rest of the paper. 
Typically, the time complexity of local MPS operations in MPS-based algorithms is  
$\mathcal{O}(\chi^{n})$, where, for instance, $n=3$ in the TEBD algorithm and the density-matrix renormalization group (DMRG) 
algorithm~\cite{white1992density,schollwock2005the}. 
In addition, the time cost of real-space sequential algorithms for simulating 
a system with $N$ sites linearly increases with $N$, as the \textit{sweep procedure} is usually required in these 
algorithms~\cite{schollwock2011the,zhou2020what}. Therefore, the complexities of MPS calculations are determined 
by the bond dimension $\chi$ and the system size $N$ for real-space sequential algorithms. These complexity analyses 
suggest that compressing an MPS to another MPS with a smaller $\chi$ while maintaining high fidelity is essential 
for accelerating MPS calculations.
Notice that we omit the dimension $d_{\rm loc}$ of the local Hilbert space in the complexity analysis 
since it is in general small (in our case, $d_{\rm loc}=2$) compared with $\chi$ and $N$.

\section{\label{sec:mps_comp}Real-space parallelizable MPS compression}

In this section, we explore the real-space parallelizable MPS compression. After introducing the concept of MPS compression, 
we assess the accuracy of straightforward parallel MPS compression (i.e., truncating all the virtual bonds simultaneously) and 
show its comparablity with a typical sequential compression method. Then, we propose a method to stabilize 
the wavefunction norm in parallel MPS compression and provide a proof for its upper and lower bounds. Next, we discuss 
the recovery of the canonical form after parallel MPS compression by additional parallel tSU steps. 
Combining these insights, we ultimately propose an improved parallel MPS compression method that is not only numerically 
more stable but also more accurate for simulations in subsequent steps.

\subsection{MPS compression}

For an MPS representing the state $|\Psi\rangle$ with bond dimension $\chi$, we can construct another MPS
\begin{equation}
    \label{eq:mps_trunc_repr}
    |\Psi_{T}\rangle = \sum_{\sigma_{1},\cdots,\sigma_{N}} \mathrm{Tr}({M'}^{[1]\sigma_{1}}{M'}^{[2]\sigma_{2}} \cdots {M'}^{[N]\sigma_{N}}) |\sigma_{1}\cdots \sigma_{N}\rangle~,
\end{equation}
with bond dimension $\chi' \leq \chi$ to approximate $|\Psi\rangle$. Here, $|\Psi_{T}\rangle$ represents the compressed 
(or truncated) MPS representation of the state $|\Psi\rangle$.

It is meaningless to compress an MPS without considering accuracy since our goal is to construct an MPS that shares 
the largest similarity with the original uncompressed one. The similarity of two MPSs is determined by the wavefunction 
fidelity between $|\Psi\rangle$ and $|\Psi_{T}\rangle$, i.e., 
\begin{equation}
    \label{eq:F_2}
    \mathcal{F}_2 = |\langle \Psi | \Psi_{T} \rangle|^{2}~.
\end{equation}
Given that $|\Psi\rangle$ and $|\Psi_{T}\rangle$ are normalized, $\mathcal{F}_2$ should always be in the region 
$0 \leq \mathcal{F}_2 \leq 1$. The accuracy of MPS compression is characterized by $\mathcal{F}_2$. 
For a given $\chi'$, our aim is to construct $|\Psi_{T}\rangle$ with the maximal $\mathcal{F}_2$. 
Furthermore, in this section, 
we assume that the MPS representing $|\Psi\rangle$ has been brought into the canonical form before being compressed.

There are two typical approaches to performing MPS compression. One involves locally and sequentially modifying each of 
the MPS tensors (\textit{sequential MPS compression}). Representative algorithms in this approach include the sequential 
SVD truncation and the sequential local variational optimization (for details, see Ref.~\cite{schollwock2011the}). 
In this case, $\mathcal{F}_2$ in Eq.~(\ref{eq:F_2}) can be evaluated from the singular values kept in $|\Psi_{T}\rangle$.
Due to the sequential nature of this approach, the time complexity of these algorithms always linearly increases 
with the system size $N$. 
The other approach involves compressing all MPS tensors simultaneously (\textit{parallel MPS compression}). 
Although this real-space parallelizable approach has a promising constant time complexity with increasing the system size $N$, 
it has been rarely discussed in the literature~\cite{verstraete2006matrix,urbanek2016parallel,paeckel2019time}. In this parallel approach, it is not straightforward to build the connection between the fidelities in each local tensor compression and the global fidelity in Eq.~(\ref{eq:F_2}).

\begin{figure}[!ht]
    \centering
    \includegraphics[width=0.8\linewidth]{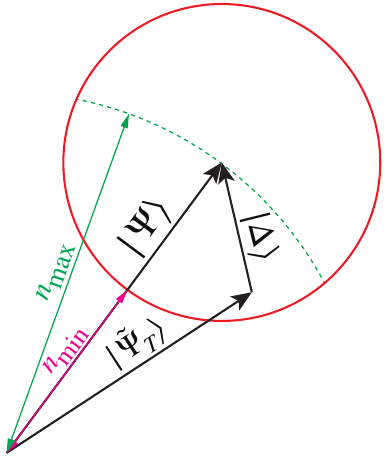}
    \caption{\label{fig:trunc_mps}
        Schematic plot illustrating the relation between the untruncated state $|\Psi\rangle$ and the state after the parallel MPS 
        compression, $|\tilde{\Psi}_{T}\rangle$, defined in Sec.~\ref{sec:thm1}. The red circle with a radius of $\sqrt{2 \epsilon(\chi')}$ 
        indicates the boundary of $|\tilde{\Psi}_{T}\rangle$ guaranteed by Theorem~\ref{thm:1}. $n_{\text{max}(\text{min})}$ 
        represents the maximum (minimum) of the wavefunction norm of $|\tilde{\Psi}_{T}\rangle$.
        } 
\end{figure}

\subsection{\label{sec:thm1}Fundamental theorem on parallel MPS compression}

A straightforward algorithm to perform parallel MPS compression (to the best of our knowledge, it is also the only existing 
algorithm) involves inserting the local projector 
$P_{i} = \sum_{\alpha = 1}^{\chi'}|\alpha\rangle_{\mathcal{V}_{i}} {}_{\mathcal{V}_{i}}\langle\alpha|$, 
where $\{|\alpha\rangle_{\mathcal{V}_{i}}\}$ are orthogonal (but not complete) bases in $\mathcal{V}_{i}$, into each bond 
in Eq.~(\ref{eq:mps_cano_decomp}) to truncate it to dimension $\chi'$. The resulting compressed MPS is denoted as 
$|\tilde{\Psi}_{T}\rangle$. It is important to note that $|\tilde{\Psi}_{T}\rangle$ is generally not normalized, while $|\Psi\rangle$ is 
always normalized. Therefore, to obtain the normalized state $|\Psi_{T}\rangle$, additional renormalization is necessary: 
$|\Psi_{T}\rangle = |\tilde{\Psi}_{T}\rangle / |\langle \tilde{\Psi}_{T}|\tilde{\Psi}_{T}\rangle|^{1/2}$.

The accuracy of the above parallel MPS compression algorithm is bounded by the follwoing fundamental 
theorem~\cite{verstraete2006matrix,cirac2021matrix}:
\begin{thm}
\label{thm:1}
\begin{equation}
    \label{eq:thm1}
    ||\Psi\rangle - |\tilde{\Psi}_{T}\rangle|^{2} \leq 2 \epsilon(\chi')
\end{equation}
and
\begin{equation}
    \label{eq:thm1_F_bound}
    \mathcal{F}_2 \geq 1 - 2 \epsilon(\chi')~,
\end{equation}
where $\epsilon(\chi') = \sum_{i = 1}^{N-1} \epsilon_{i}(\chi')$ is defined as the global truncation error, and $\epsilon_{i}(\chi') = \sum_{\alpha = \chi'+1}^{\chi}(\Lambda^{[i]}_{\alpha \alpha})^{2}$ is the local truncation error associated with bond $i$. 
An illustration of Eq.~(\ref{eq:thm1}) is shown in Fig.~\ref{fig:trunc_mps}.
\end{thm}

Combining Eq.~(\ref{eq:thm1_F_bound}) with the fact that $0 \leq \epsilon_{i} \leq 1$, we can derive a looser bound for 
$\mathcal{F}_2$ as
\begin{equation}
    \label{eq:thm1_F_bound_looser}
    \mathcal{F}_2 \geq 1 - 2 \sum_{i=1}^{N-1} \sqrt{\epsilon_{i}(\chi')}~.
\end{equation}
Note that the right-hand side of Eq.~(\ref{eq:thm1_F_bound_looser}) can be negative, 
while the wavefunction fidelity $\mathcal{F}_2$ must be positive.
The bound in Eq.~(\ref{eq:thm1_F_bound_looser}) has the same formulation as the exact fidelity lower bound for the sequential SVD 
truncation formulated in Ref.~\cite{zhou2020what}. Although $\epsilon_{i}(\chi_i)$ may slightly deviates from the corresponding local 
truncation 
error in the sequential SVD compression, we argue that the parallel MPS compression algorithm could have the same order of accuracy 
as the sequential SVD compression algorithm. This argument will be numerically confirmed in Sec.~\ref{sec:ptebd_accu}.
Also notice that, in general, the canonical form is broken after a parallel MPS compression, 
which could impact on the subsequent simulation steps.

\subsection{Wavefunction norm stabilization in parallel MPS compression}

Although the norm of the wavefunction does not impact the calculation of any physical observables, 
it strongly influences numerical stability in actual numerical simulations, as it may vanish during the calculation process. 
Hence, we have to carefully consider the possible norm deviation induced by the parallel MPS compression.

\subsubsection{Decay of the wavefunction norm after parallel MPS compression}

To evaluate the wavefunction norm deviation after the parallel MPS compression, we first derive its bounds in Lemma~\ref{thm:2}. 
Its proof is based on Theorem~\ref{thm:1} and the properties of the canonical form.
\begin{lemma}
\label{thm:2}
The wavefunction norm $n = ||\tilde{\Psi}_{T}\rangle|$ after the parallel MPS compression is bounded by
\begin{equation}
    \label{eq:mps_gt_n_bounds}
    1 - \sqrt{2\epsilon(\chi')} \leq n \leq 1~.
\end{equation}
\end{lemma}

\begin{proof}
    According to the triangle inequality, we have $1 - ||\Delta\rangle| \leq n$, where $|\Delta\rangle = |\Psi\rangle - |\tilde{\Psi}_{T}\rangle$, 
    as schematically illustrated in Fig.~\ref{fig:trunc_mps}. Then, we can derive the lower bound in Eq.~(\ref{eq:mps_gt_n_bounds}) by using Theorem~\ref{thm:1} as
    \begin{equation}
        1 - \sqrt{2\epsilon(\chi')} \leq 1 - ||\Delta\rangle| \leq n~.
    \end{equation}
    
    As for the upper bound in Eq.~(\ref{eq:mps_gt_n_bounds}), from the direct calculation, we have
    \begin{eqnarray}
        \label{eq:n2}
        \nonumber
        n^{2} = &&\sum_{\sigma_1,\cdots,\sigma_N}\text{Tr}(A^{[1]\sigma_1} P_{1} \cdots P_{N-2} A^{[N-1]\sigma_{N-1}} P_{N-1} A^{[N]\sigma_N} \\
        &&\bar{A}^{[N]\sigma_N} P_{N-1} \bar{A}^{[N-1]\sigma_{N-1}} P_{N-2} \cdots P_{1} \bar{A}^{[1]\sigma_1})~.
    \end{eqnarray}
    Here, we adopt the left canonical form described in Eq.~(\ref{eq:l_r_cano}). By defining the quantum map
    \begin{equation}
      \mathcal{E}^{[i]} (X) = \sum_{\sigma_{i}} A^{[i]\sigma_{i}} X \bar{A}^{[i]\sigma_{i}}~,
    \end{equation}
    where $X$ is a $(\mathcal{V}_{i}\times\mathcal{V}_{i})$ matrix, 
    Eq.~(\ref{eq:n2}) can be reformulated as
    \begin{eqnarray}
        \label{eq:n2_2}
        \nonumber
        n^{2} = &&\mathrm{Tr}[ \mathcal{E}^{[1]} ( P_{1} \cdots \mathcal{E}^{[N-2]} ( P_{N-2}\\
        &&\mathcal{E}^{[N-1]}(P_{N-1} \rho^{[N]} P_{N-1})P_{N-2} ) \cdots P_{1} ) ]
    \end{eqnarray}
    and $\rho^{[N]}\equiv \sum_{\sigma_N}A^{[N]\sigma_N}\bar{A}^{[N]\sigma_N}$ is obviously a positive matrix. 
    Introducing the following recursive relation: 
    \begin{equation}
        \label{eq:recu}
        Y^{[i]} = \mathcal{E}^{[i]}(P_{i} Y^{[i+1]} P_{i})~\text{with}~Y^{[N]} = \rho^{[N]}
    \end{equation}
    for $i=1,2, \cdots, N-1$, Eq.~(\ref{eq:n2_2}) can be simply given as $n^{2} = \text{Tr}(Y^{[1]})$. 

    Next, we define another map
    \begin{equation}
        \mathcal{P}_{i}(X) = P_{i}XP_{i}~.
    \end{equation}
    Because $\mathcal{E}^{[i]}(X)$ is a completely positive trace preserving map~\cite{choi1975completely} 
    and $\mathcal{P}_{i}(X)$ is a positive map, $Y^{[i]}$ is a positive matrix. We can also readily show that 
    \begin{equation}
        \mathrm{Tr}(X) = \mathrm{Tr}(\mathcal{P}_{i}(X)) + \mathrm{Tr}(\bar{\mathcal{P}}_{i}(X))~,
    \end{equation}
    where $\bar{\mathcal{P}}_{i}(X) = (I-P_{i})X(I-P_{i})$ is also a positive map. 
    Since $\mathrm{Tr}(\bar{\mathcal{P}}_{i}(X)) \geq 0$ for a positive matrix $X$, 
    $\mathrm{Tr}(\mathcal{P}_{i}(X)) \leq \mathrm{Tr}(X)$. 
    Noticing that $Y^{[i]}$ is a positive matrix, we can finally obtain the following inequality: 
    \begin{equation}
        \label{eq:Y_ineq}
      \mathrm{Tr}(Y^{[i]}) \leq \mathrm{Tr}(Y^{[i+1]})~.
    \end{equation}
    Using this inequality, we can derive the upper bound in Eq.~(\ref{eq:mps_gt_n_bounds}) as
    \begin{equation}
        \label{eq:n2_upper}
        n = \sqrt{\text{Tr}(Y^{[1]})} \leq \sqrt{\text{Tr}(Y^{[N]})} = 1~.
    \end{equation}
\end{proof}

Lemma~\ref{thm:2} shows that the wavefunction norm necessarily decreases, except for the ideal case, i.e., without truncating any bonds, 
after a parallel MPS compression process. 
This implies that, if we alternatively perform TEBD procedures and parallel MPS compression processes, as in the case of 
Ref.~\cite{urbanek2016parallel}, the wavefunction norm will monotonically decay, ultimately leading to serious numerical 
instability. This norm vanishing will be numerically demonstrated later in Sec.~\ref{sec:ptebd_stab}. 
To avoid the numerical instability, 
naively, we can renormalize the wavefunction after a parallel MPS compression process. 
However, the time cost of calculating 
the wavefunction norm increases linearly with the system size $N$, offsetting the benefits of parallel MPS compression. 
Therefore, it is highly desired to develop an efficient method that can suppress the norm vanishing during repeated 
parallel MPS compression processes without any sequential procedure.

\subsubsection{\label{sec:thm2}Parallelizable wavefunction norm stabilization}

To stabilize the wavefunction norm, a straightforward possibility is to locally multiply $\Lambda^{[i]}$ 
by a factor $\nu_{i} = [1 - \epsilon_{i}(\chi')]^{-1/2}$ after a parallel MPS compression process. 
Consequently, the wavefunction norm changes to
\begin{equation}
    n^{*} = n\prod_{i=1}^{N-1}\nu_{i}~.
    \label{eq:lren}
\end{equation}
Now, we shall prove that $n^{*}$, denoted as stabilized norm, is bounded in an interval that includes 1, 
and the two boundaries of this interval uniformly converge to 1 as $\epsilon(\chi')$ approaches 0.

\begin{thm}
\label{thm:3}
The stabilized norm $n^{*}$ is bounded by
\begin{eqnarray}
    n^{*}_{\text lower} &=& (1 - \sqrt{2\epsilon(\chi')})\prod_{i=1}^{N-1}\nu_{i}~, \label{eq:n_low}\\
    n^{*}_{\text upper} &=& \prod_{i=1}^{N-1}\nu_{i}~, \label{eq:n_up}
\end{eqnarray}
with $n^{*}_{\text lower} \leq 1$ and $n^{*}_{\text upper} \geq 1$. 
Moreover, they uniformly converge to 1 as $\epsilon(\chi') \rightarrow 0$.
\end{thm}

\begin{proof}
Eqs.~(\ref{eq:n_low}) and (\ref{eq:n_up}) follow directly from Eq.~(\ref{eq:mps_gt_n_bounds}) in Lemma~\ref{thm:2}. 
Let us now prove that these quantities satisfy $n^{*}_{\text lower} \leq 1$ and $n^{*}_{\text upper} \geq 1$. 
Noticing that $\epsilon(\chi') \geq 0$ and $1 - 2\epsilon(\chi') = (1 + \sqrt{2\epsilon(\chi')}) (1 - \sqrt{2\epsilon(\chi')})$, 
we have the following inequalities:
\begin{equation}
    \label{eq:thm3_1}
    1 - \sqrt{2\epsilon(\chi')} \leq 1 - 2\epsilon(\chi') \leq 1 - \epsilon(\chi')~,
\end{equation}
when $2\epsilon(\chi')\leq1$. 
On the other hand, since $0 \leq \epsilon_{i}(\chi') \leq 1$, we can show that 
\begin{equation}
    \label{eq:thm3_2}
    \prod_{i = 1}^{N-1} \nu_{i}^{-1} \geq \prod_{i = 1}^{N-1}[1 - \epsilon_{i}(\chi')] \geq 1 - \epsilon(\chi').
\end{equation}
Combine Eq.~(\ref{eq:thm3_1}) and Eq.~(\ref{eq:thm3_2}), we can obtain that 
$1 - \sqrt{2\epsilon(\chi')} \leq \prod_{i = 1}^{N-1} \nu_{i}^{-1}$ and thus 
$n^{*}_{\text lower} = (1 - \sqrt{2\epsilon(\chi')})\prod_{i=1}^{N-1}\nu_{i} \leq 1$. 
This is trivial when $2\epsilon(\chi')\geq1$. 
It is also straightforward to show that $n^{*}_{\text upper} \geq 1$ because $\epsilon_{i}(\chi') \leq 1$. 

Since both quantities $1 - \sqrt{2\epsilon(\chi')}$ and $\prod_{i=1}^{N-1}\nu_{i}$ converge to 1 
as $\epsilon(\chi') \rightarrow 0$, the two boundaries $n^{*}_{\text lower}$ and $n^{*}_{\text upper}$ converge to 1 from 
both sides. 
\end{proof}

Theorem~\ref{thm:3} implies that the local renormalization method can indeed effectively mitigate the decay of the wavefunction norm 
during the parallel MPS compression. While the proof provided above assumes that the uncompressed MPS is in its canonical form, 
we argue that this wavefunction norm stabilization method remains applicable even when the uncompressed MPS deviates slightly 
from the canonical form. The numerical demonstration supporting this assertion will be presented in Sec.~\ref{sec:ptebd_stab}.

\subsection{Parallel recovery of the canonical form through parallel tSU}

Because of its many advantages, almost all MPS-based algorithms developed previously employ the canonical form. 
Therefore, it is important to maintain an MPS in the canonical form or close to the canonical form also in a parallel MPS algorithm. 
In this respect, it is well-known that the right-left sweeps of tSU, which is a sequential procedure, can converge an arbitrary 
MPS to its canonical form. However, how the canonical form is recovered with a parallel procedure has been rarely investigated 
thus far. 
In this section, we first prove that
at most $\frac{N}{2}$ ($\frac{N-1}{2}$) parallel tSU steps
can converge an arbitrary MPS to its canonical form when $N$ is even (odd), and then we discuss how it approaches the canonical form with fewer 
operations.

Similar to the TEBD algorithm, which is exactly parallelizable in the case of no truncation~\cite{vidal2003efficient}, 
we can define a single parallel trivial simple update (PtSU) step as applying tSU on all the odd MPS bonds and then 
on all the even MPS bonds. 
First, we prove the convergency of an MPS to its canonical form through PtSU steps with the following theorem.

\begin{thm} \label{thm:4}
For an MPS with $N$ sites, at most $\frac{N}{2}$ ($\frac{N-1}{2}$) PtSU steps are required 
to convert it to its canonical form when $N$ is even (odd).
\end{thm}

\begin{proof}
Before presenting the proof, let us first introduce an equivalent definition of the canonical form. 
We define the following left tensor $L^{[i]}$ and right tensor $R^{[i]}$ for bond $i$, respectively: 
\begin{align}
    L^{[i]\sigma_1\cdots\sigma_i} \equiv& \Lambda^{[0]}\Gamma^{[1]\sigma_1}\cdots\Lambda^{[i-1]}\Gamma^{[i]\sigma_i} \\
    R^{[i]\sigma_{i+1}\cdots\sigma_N} \equiv& \Gamma^{[i+1]\sigma_{i+1}}\Lambda^{[i+1]}\cdots\Gamma^{[N]\sigma_N}\Lambda^{[N]}~.
\end{align}
Then, asserting that an MPS composed of $\{\Gamma^{[i]\sigma_i}, \Lambda^{[j]}\}$ is in the canonical form, 
satisfying the conditions in Eq.~(\ref{eq:l_r_cano}), is equivalent to saying 
that $L^{[i]\sigma_1\cdots\sigma_i}$ and $R^{[i]\sigma_{i+1}\cdots\sigma_N}$ satisfy 
\begin{eqnarray}
    \label{eq:llid}
    \sum_{\sigma_1,\cdots,\sigma_i}\bar{L}^{[i]\sigma_1\cdots\sigma_i}L^{[i]\sigma_1\cdots\sigma_i} &=& I_{\mathcal{V}_{i}}\\
    \label{eq:rrid}
    \sum_{\sigma_{i+1},\cdots,\sigma_N} R^{[i]\sigma_{i+1}\cdots\sigma_N}\bar{R}^{[i]\sigma_{i+1}\cdots\sigma_N} &=& I_{\mathcal{V}_{i}}
\end{eqnarray}
for $\forall i \in [0, N]$, where $L^{[0]}\equiv(1)$ and $R^{[N]}\equiv(1)$. 
The proof of this equivalence follows from straightforward calculations.

Using this equivalent definition of the canonical form, we can prove the theorem inductively. 
As the starting point, $\bar{L}^{[0]}L^{[0]} = (1) = I_{\mathcal{V}_{0}}$ is trivially satisfied. 
Then, if we assume that $\sum_{\sigma_1,\cdots,\sigma_{i-1}}\bar{L}^{[i-1]\sigma_{1}\cdots\sigma_{i-1}}L^{[i-1]\sigma_1\cdots\sigma_{i-1}} = I_{\mathcal{V}_{i-1}}$ is satisfied, 
one additional tSU on sites $i$ and $i+1$ can realize 
$\sum_{\sigma_1,\cdots,\sigma_{i}}\bar{L}^{[i]\sigma_{1}\cdots\sigma_{i}}L^{[i]\sigma_1\cdots\sigma_{i}} = I_{\mathcal{V}_{i}}$ 
This is simply because, after this tSU, $\Lambda^{[i-1]}\Gamma^{[i]\sigma_i} = A^{[i]\sigma_i}$ satisfies 
$\sum_{\sigma_i}\bar{A}^{[i]\sigma_i}A^{[i]\sigma_i} = I_{\mathcal{V}_{i}}$ and thus  
\begin{align}
&\sum_{\sigma_1,\cdots,\sigma_{i}}\bar{L}^{[i]\sigma_{1}\cdots\sigma_{i}}L^{[i]\sigma_1\cdots\sigma_{i}} \nonumber \\
&= \sum_{\sigma_1,\cdots,\sigma_{i}}\bar{A}^{[i]\sigma_i} \bar{L}^{[i-1]\sigma_{1}\cdots\sigma_{i-1}}L^{[i-1]\sigma_{1}\cdots\sigma_{i-1}} A^{[i]\sigma_i} = I_{\mathcal{V}_{i}}.
\end{align} 
At the same time, by definition, any tSU involving sites in $\{1,\cdots,i\}$ does not change the left tensors 
$\{L^{[k]\sigma_1\cdots\sigma_{k}}\}_{k=0}^{i-1}$. Therefore, $\bar{L}^{[N-1]}L^{[N-1]} = I_{\mathcal{V}_{N-1}}$ is established 
after $\frac{N}{2}$ ($\frac{N-1}{2}$) PtSU steps for even (odd) $N$. Note that Eq.~(\ref{eq:llid}) is satisfied for $i=N$ because of 
the wavefunction normalization condition.

Similarly, we can prove that the same PtSU steps also establish Eq.~(\ref{eq:rrid}) for $\forall i \in [1, N-1]$. 
For $i = 0$, Eq.~(\ref{eq:rrid}) is automatically satisfied because of the wavefunction normalization condition. 
Finally, since this proof is inductive, it provides an upper bound on how many PtSU steps are required to restore the canonical form.
\end{proof}

Theorem~\ref{thm:4} shows that the PtSU, a real-space parallelizable algorithm, can recover the canonical form with only locally 
operating tensors, implying the potential of its utilization on modern supercomputers where communication between 
neighboring nodes is highly optimized~\cite{ajima2018tofu}.
Moreover, we can simply skip applying the identity gate in the tSU procedure to further optimize the performance in practice.

Now, we explore the behavior of approaching the canonical form by applying the PtSU.
To measure how an MPS approaches its canonical form with PtSU steps, we first define the following distance to the canonical form: 
\begin{equation}
  \label{eq:dist_to_cano}
   \mathcal{C} = \frac{1}{2N} \sum_{i=1}^{N} ||\bar{A}^{[i]}A^{[i]} - I_{\mathcal{V}_{i}} || + || B^{[i]}\bar{B}^{[i]} - I_{\mathcal{V}_{i-1}} ||~,
\end{equation}
where $||\cdot||$ represents the Frobenius norm of the matrix. This quantity is zero if and only if an MPS is in the canonical form. Note that, in Eq.~({\ref{eq:dist_to_cano}}), we assume that the norm of the MPS is 1.
Notice that a similar quantity has been proposed in a recent study 
to examine the regauging of general tensor networks, where the nuclear norm is considered~\cite{tindall2023gauging}.

We investigate how $\mathcal{C}$ changes with the application of PtSU steps. To closely mimic the actual calculation task, 
we prepare the initial MPS as follows: 
We first construct a canonicalized random MPS with the bond dimension $\chi$. Each tensor element in this MPS is initialized 
to a complex number $z = a + ib$ with $a$ and $b$ chosen randomly from $[-1, 1]$,  
and then its canonical form is constructed with the norm properly normalized to 1.
Subsequently, we truncate all the virtual bonds to $\chi/2$, i.e., 
disregarding the $\chi/2$ smallest elements in each $\Lambda^{[i]}$, and renormalize the wavefunction norm to 1.
We then perform PtSU steps on this MPS and calculate $\mathcal{C}$ after each PtSU step. The results are summarized 
in Fig.~\ref{fig:Cs_vs_ptsu}.
Here, we choose $N=20$ and 24 and examine the cases of $\chi = 32$, 128, 256 and 512.

\begin{figure}[!ht]
    \centering
    \includegraphics[width=\linewidth]{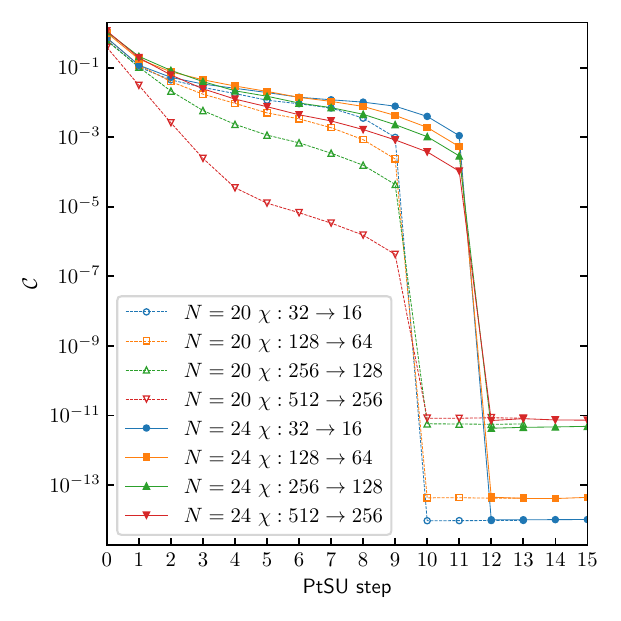}
    \caption{\label{fig:Cs_vs_ptsu}
    Seim-log plot of $\mathcal{C}$
    as a function of PtSU steps for $N=20$ and 24. 
    The initial MPS is prepared by simultaneously truncating all the virtual bonds of a canonicalized random MPS to half of the bond 
    dimension and manually normalizing the norm of the resulting MPS to 1. 
    The results are averaged over 100 initial MPSs with different sets of random parameters.}
\end{figure}

In Fig.~\ref{fig:Cs_vs_ptsu}, we observe that $\mathcal{C}$ rapidly and monotonically decreases with PtSU steps, especially 
at the beginning of several PtSU steps. Additionally, the canonical form is attained more rapidly when the bond dimension is larger. 
This implies that we can drive the MPS closer to the canonical form 
by appending a few PtSU steps after the parallel MPS compression.
Also notice that the canonical form is exactly restored after $N/2$ PtSU steps, justifying Theorem~\ref{thm:4}.

\subsection{\label{sec:ipmc}Improved parallel MPS compression (IPMC)}

Inspired by the new insights from Theorems~\ref{thm:3}~and~\ref{thm:4}, we propose the IPMC method: starting from an MPS 
in the Vidal form, which is usually the resulting MPS of the previous TEBD evolution and has the bond dimension $d_{\rm loc}\chi$ 
($d_{\rm loc}$ is the dimension of the local Hilbert space, i.e., $d_{\rm loc}=2$ for the qubit case),
we perform a parallel MPS compression process (discussed in Sec.~\ref{sec:thm1}) followed by a norm stabilization step 
(discussed in Sec.~\ref{sec:thm2}) to compress the bond dimension back to $\chi$. Then, we perform $g$ PtSU steps to 
partially recover the canonical form. This parallel MPS compression method is fully real-space parallelizable, requiring 
only communication between neighboring sites [also see Fig.~\ref{fig:ptebd_algo}(d)]. 
In the next section, we will apply this method to simulate unitary 
quantum dynamics and examine its efficiency through extensive numerical experiments.

\section{\label{sec:app_uqd}Application to the simulation of unitary quantum dynamics}

The time evolution operator $\hat{\mathcal{U}}(t)$, which describes the unitary quantum dynamic process of evolving a state 
from time $0$ to $t$, $|\Psi(t)\rangle = \hat{\mathcal{U}}(t)|\Psi(0)\rangle$, can always be decomposed into a quantum circuit composed of layers of local unitary operators
\begin{equation}
    \label{eq:circ_decomp}
    \hat{\mathcal{U}}(t) = \prod_{d = 1}^{D} \prod_{i = 1}^{m(d)} \hat{U}^{d}_{i},
\end{equation}
where $\hat{U}_{i}^{d}$ is the $i$th local unitary operator in the $d$th layer, $m(d)$ is the number of local unitary operators 
in the $d$th layer, and $D$ is the total number of layers (i.e., depth) of the quantum circuit~\cite{shirakawa2021automatic}.
Notice that there exist other simulation approaches when $\hat{\mathcal{U}}(t)$ is associated with a Hamiltonian \cite{paeckel2019time,secular2020parallel}.

\begin{figure*}[!ht]
    \centering
    \includegraphics[width=\linewidth]{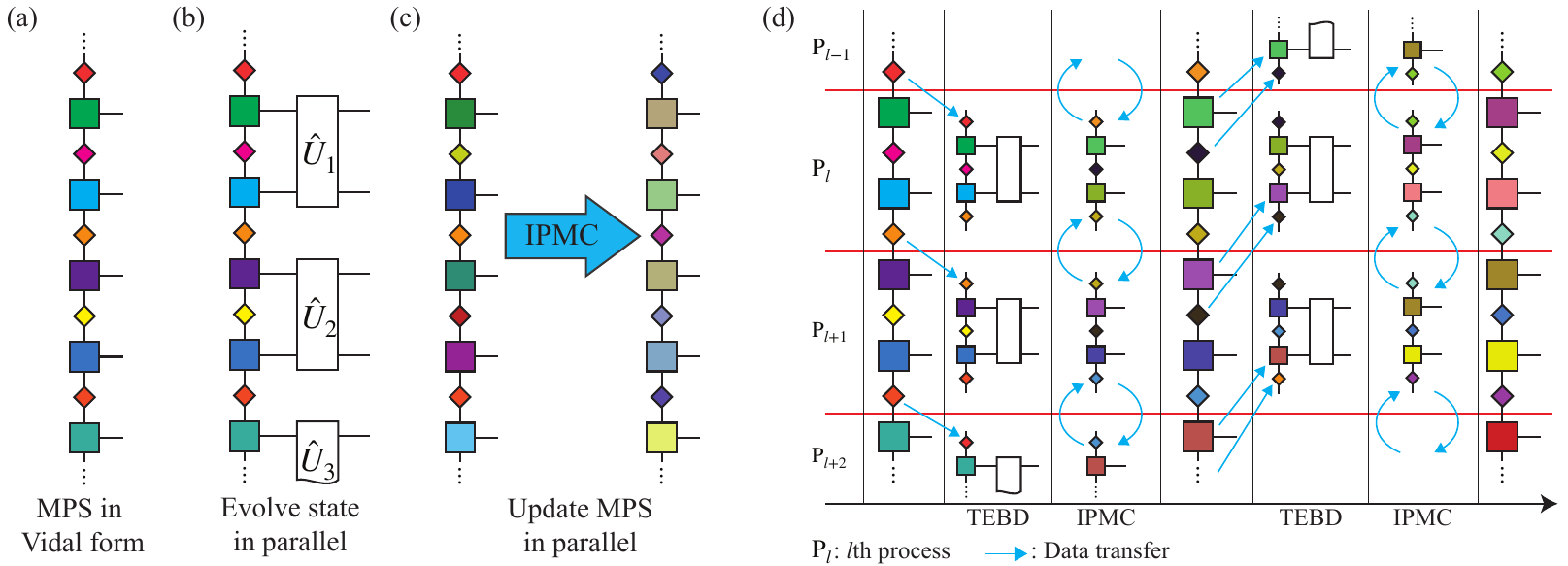}
    \caption{\label{fig:ptebd_algo}
        Illustration of the pTEBD algorithm for a quantum circuit simulation.
        Starting from an MPS in (a), which is allowed to be slightly deviated from the 
        canonical form, the state is evolved using the standard TEBD procedure but with all gates in the same layer treated parallelly 
        in (b), and then truncate the state using the IPMC (see Sec.~\ref{sec:ipmc}) in (c) to update the MPS, which is 
        subsequently used to simulate the quantum circuit with the next layer of gates. 
        A possible realization of the distribution of data to each MPI process ${\rm P}_l$ and the data transfer among MPI processes in the pTEBD algorithm 
        is schematically shown in (d). As multiple PtSU steps may be executed in the IPMC, the data transfer is consequently 
        required several times, as indicated here simply by two arrows pointing forward and backward between neighboring MPI 
        processes. 
        Note also that single-qubit gates can be treated exactly without increasing the bond dimension by any MPS-based simulations, 
        including the pTEBD algorithm. 
        }
\end{figure*}

The above decomposition is, in general, approximate and, moreover, it is not unique. Therefore, for the convenience of 
applying MPS-based algorithms, one usually constructs $\{\hat{U}_{i}^{d}\}$ only involving adjacent qubits in a 1D path 
going through all the qubits in the system [see Fig.~\ref{fig:ptebd_algo}(b) as an example and also 
Appendix~\ref{app:qc_setup}]. Therefore, in the rest of this section, we discuss the simulations of unitary quantum dynamics by 
considering their 1D quantum circuit representations.

\subsection{\label{sec:ptebd}Parallel time-evolving block decimation (pTEBD) algorithm }

A unitary quantum dynamic process represented by a 1D quantum circuit with $N$ qubits can be efficiently simulated 
using a sequential TEBD-type \cite{zhou2020what} or time-dependent DMRG-type \cite{ayral2023density} algorithm. 
The time cost of these two algorithms is always proportional to $N$. 
Here, incorporating the IPMC method introduced above, 
we propose a parallel TEBD (pTEBD) algorithm that can reduce the time complexity to a constant with increasing $N$.

The pTEBD algorithm along with the IPMC is illustrated in Fig.~\ref{fig:ptebd_algo}. 
The initial state $|\Psi(0)\rangle$ can be represented as a canonicalized MPS with bond dimension $\chi_{0}$ [see Fig.~\ref{fig:ptebd_algo}(a)]. To simulate quantum circuits for a circuit-based quantum computer, generally one chooses 
$|\Psi(0)\rangle = |00\cdots 0\rangle$ with $\chi_{0} = 1$. 
In the pTEBD algorithm, to simulate one layer of gates represented by $\prod_{i}^{m(d)} \hat{U}^{d}_{i}$ in Eq.~(\ref{eq:circ_decomp}), 
the standard TEBD procedure is utilized to simulate each gate operation in parallel [see Fig.~\ref{fig:ptebd_algo}(b)]. 
The bond dimension of the resulting MPS increases to $\tilde{\chi}$, which can be larger than the maximum bond dimension 
$\chi$ predetermined in the algorithm. 
Then, the IPMC is performed to truncate the increased bond dimension back to $\chi$ [see Fig.~\ref{fig:ptebd_algo}(c)] 
and it returns to the first step to simulate the next layer of gates in the quantum circuit. 
As indicated in Fig.~\ref{fig:ptebd_algo}(d), the data transfer in the pTEBD algorithm is strictly local.

After completing the simulation of all $D$ layers of gates in the quantum circuit, the resulting MPS approximates the final 
state of the unitary quantum dynamic process $|\Psi(t)\rangle$.  We can then use standard MPS techniques to exactly 
evaluate the expectation values of observables \cite{schollwock2011the} or perfectly sample bit strings 
from this state \cite{ferris2012perfect} to mimic quantum devices after restoring the canonical form of the obtained MPS representing 
$|\Psi(t)\rangle$. 
The flow chart of the pTEBD algorithm is provided in Algorithm~\ref{alg:ptebd}.

\begin{figure*}[!ht]
    \centering
    \includegraphics[width=\linewidth]{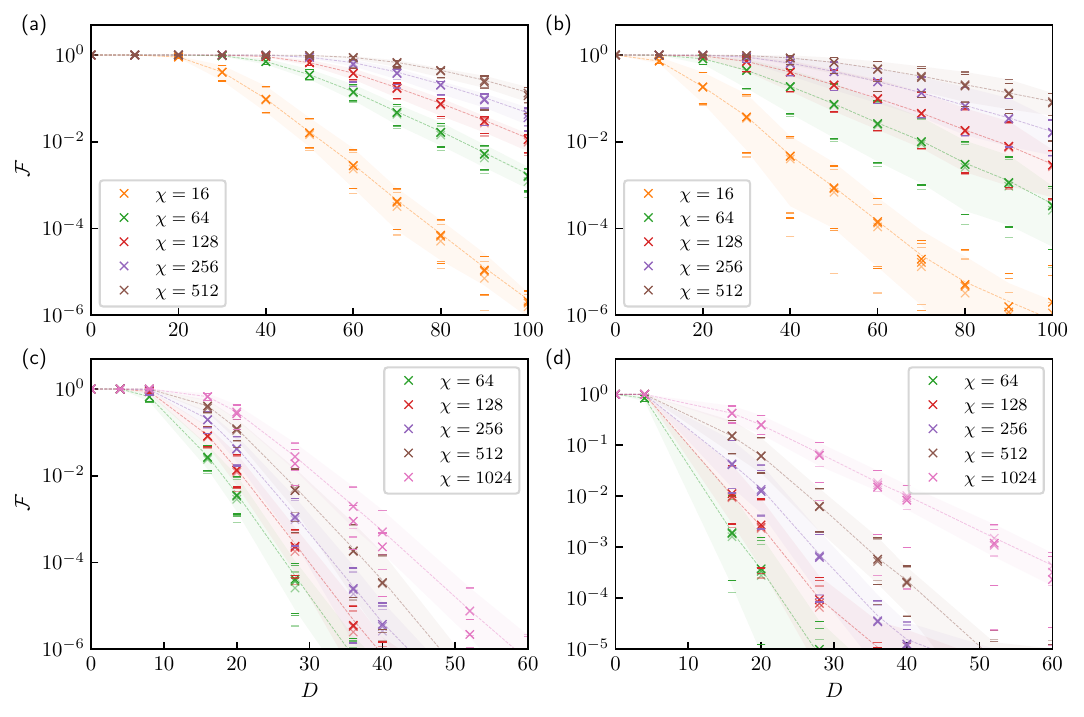}
    \caption{\label{fig:Fs_vs_D}
        Wavefunction fidelity $\mathcal{F}$ as a function of physical circuit depth $D$ 
        (for the definition, see Appendix~\ref{app:qc_setup})
        obtained using the sequential MPS algorithm 
        and the pTEBD algorithm with a fixed MPS bond dimension $\chi$:
        (a) RQC-1D with $N = 25$, (b) PQC-1D with $N = 24$, (c) RQC-2D with $L_{x} = 5$ and $L_{y} = 5$ (i.e., $N = 25$), 
        and (d) PQC-2D with $L_{x} = 4$ and $L_{y} = 6$ (i.e., $N = 24$).
        The dashed lines (crosses) represent the average fidelity $\bar{\mathcal{F}}_{\rm MPS (pTEBD)}$ over 10 simulations 
        of the same quantum circuit but with different sets of random parameters obtained using the sequential MPS 
        (pTEBD) algorithm.
        The shades (bars) indicate the minimal and maximal fidelities among these 10 simulations using the sequential 
        MPS (pTEBD) algorithm.
        The color intensities of crosses and bars represent the pTEBD results with $g = 0,1, \text{and}~2$ (from light to dark).
        }
\end{figure*}

\begin{algorithm}[H]
\caption{pTEBD}\label{alg:ptebd}
\begin{algorithmic}
    \State \textbf{Input:} $|\Psi(d=0)\rangle = \{\Gamma^{[p]}, \Lambda^{[q]}\}$, $\{\hat{U}^{d}_{i}\}$, $\chi$, $g$
    \For{$m = 1, \cdots, D$}
        \State $\{\tilde{\Gamma}^{[p]}, \tilde{\Lambda}^{[q]}\} \gets \textbf{TEBD}(\{\Gamma^{[p]}, \Lambda^{[q]}\}, \{\hat{U}^{m}_{i}\})$ \Comment{In parallel}
        \If{$\tilde{\chi} > \chi$}
            \State $\{\Gamma^{[p]}, \Lambda^{[q]}\} \gets \textbf{IPMC}(\{\tilde{\Gamma}^{[p]}, \tilde{\Lambda}^{[q]}\}, \chi, g)$
        \Else
            \State $\{\Gamma^{[p]}, \Lambda^{[q]}\} \gets \{\tilde{\Gamma}^{[p]}, \tilde{\Lambda}^{[q]}\}$
        \EndIf
    \EndFor
    \State $\{\Gamma^{[p]}, \Lambda^{[q]}\} \gets \textbf{Canonicalize}(\{\Gamma^{[p]}, \Lambda^{[q]}\})$
    \State $|\Psi(d=D)\rangle \gets \{\Gamma^{[p]}, \Lambda^{[q]}\}$
    \State \textbf{Do} measurements or sampling on $|\Psi(d=D)\rangle$
\end{algorithmic}
\end{algorithm}

\subsection{\label{sec:ptebd_2}Accuracy, numerical stability, and performance of pTEBD}

To comprehensively benchmark the accuracy, numerical stability, and performance of the pTEBD algorithm, we perform intensive 
simulations of typical 1D and 2D quantum circuits and compare the results with those obtained by sequential MPS simulations. 
For this purpose, we consider both random quantum circuits (RQCs) on 1D and 2D qubit arrays (RQC-1D and RQC-2D, respectively) 
and parametrized quantum circuits (PQCs), which are widely adopted in the variational quantum 
eigensolver (VQE)~\cite{TILLY2022}, on 1D and 2D qubit arrays (PQC-1D and PQC-2D, respectively). 
These four types of quantum circuits cover representative cases, 
providing concrete demonstrations for the practical feasibility of the pTEBD algorithm. The details of each circuit are described 
in Appendix~\ref{app:qc_setup}. 
Note also that in the sequential algorithm described in Ref.~\cite{zhou2020what},
which we employ as the benchmark algorithm, referred to as the \textit{sequential MPS algorithm},
in the following numerical experiments,
an MPS is always in the canonical form and normalized 
to one, which is generally not the case in the pTEBD algorithm. 
We also provide an additional benchmark test of the pTEBD algorithm applied to simulate the quantum Fourier transformation algorithm in Appendix~\ref{app:QFT_pTEBD}.

\begin{figure*}[!ht]
    \centering
    \includegraphics[width = \linewidth]{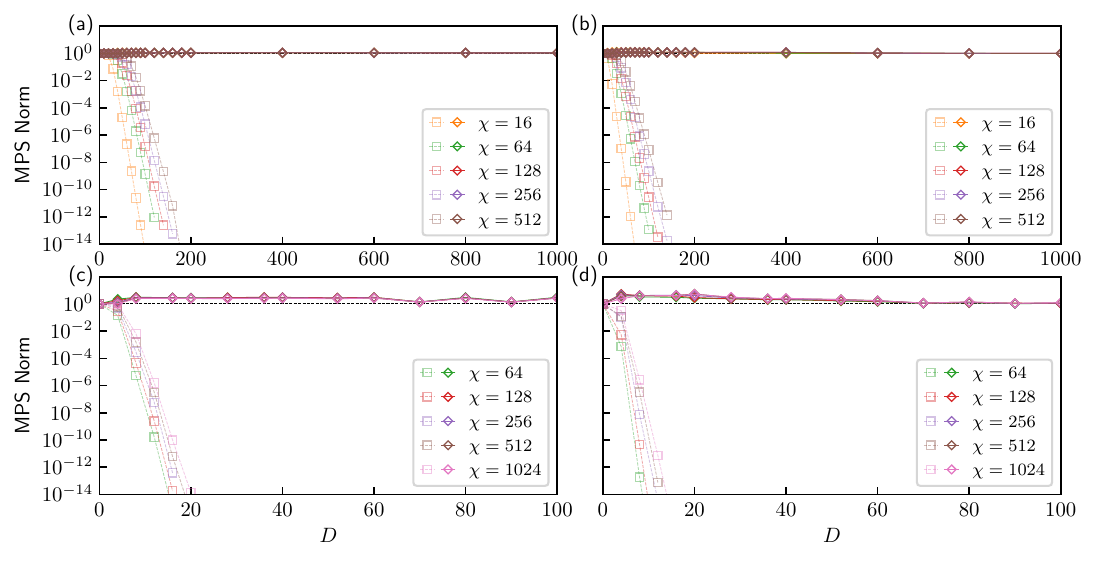}
    \caption{\label{fig:MPSNorms_vs_D}
    Wavefunction norm of the MPS as a function of physical circuit depth $D$ (for the definition, see Appendix~\ref{app:qc_setup}) 
    obtained using the pTEBD algorithm with a fixed MPS 
    bond dimension $\chi$: (a) RQC-1D with $N = 101$, (b) PQC-1D with $N = 100$, (c) RQC-2D with $L_{x} = L_{y} = 12$ 
    (i.e., $N = 144$), and (d) PQC-2D with $L_{x} = L_{y} = 12$ (i.e., $N = 144$). The open diamonds (squares) represent the averaged 
    results of the wavefunction norm $n^*$ ($n$) over 10 simulations of the same quantum circuit but with different sets of random 
    parameters, evaluated  
     with (without) the wavefunction norm stabilization procedure. Black dashed lines indicate the norm equal to 1. 
     Here we set $g=0$, i.e., no PtSU since it does not affect the wavefunction norm. 
     Note that the results obtained using the pTEBD algorithm 
     with the wavefunction norm stabilization procedure
     are almost completely overlapped to each other in this scale. 
     }
\end{figure*}

\subsubsection{\label{sec:ptebd_accu}Accuracy}

First of all, we assess the simulation accuracy of the pTEBD algorithm and demonstrate its ability to achieve precision comparable 
to that of the sequential MPS algorithm. Here, the simulation precision is determined by 
the wavefunction fidelity
\begin{equation}
    \mathcal{F} = |\langle \Psi_{\rm exact} | \Psi(\chi) \rangle|^{2}~,
    \label{eq:F}
\end{equation}
where $|\Psi_{\rm exact}\rangle$ represents the state obtained by the exact simulation and $| \Psi(\chi) \rangle$ is the state obtained 
using an MPS with the fixed bond dimension $\chi$. We perform the simulations for quantum circuits with up to 100 (60) 
physical circuit layers,
i.e., before recompiling a quantum circuit to fit an MPS 1D path,
in the 1D (2D) cases.
The definitions of the \textit{physical} and \textit{compiled} circuit depth can be found in Appendix~\ref{app:qc_setup}.
Each simulation is repeated 10 times with different sets of random parameters to evaluate  
the average fidelity, denoted as $\bar{\mathcal{F}}_{\rm MPS (pTEBD)}$ in the sequential MPS (pTEBD) simulations. 
These results are summarized in Fig.~\ref{fig:Fs_vs_D}.

The simulations for all these four types of quantum circuits show consistent results: for a given bond dimension $\chi$, 
the accuracy of the pTEBD simulation is comparable to the sequential MPS simulation.
In most cases, we find that $\bar{\mathcal{F}}_{\rm pTEBD} \approx \bar{\mathcal{F}}_{\rm MPS}$ 
and $\bar{\mathcal{F}}_{\rm pTEBD}$ improves with an increase in the number $g$ of PtSU steps, 
while we observe that the PtSU steps have the opposite effect in some rare cases.
Moreover, in cases with relatively small $\chi$ (for example, $\chi = 16$ and 64), we notice that 
$\bar{\mathcal{F}}_{\rm pTEBD}$ significantly increases with applying more PtSU steps.
Since $\mathcal{F}$ is highly sensitive to the random parameters used in each quantum gate in the quantum circuit, 
we also shown its minimum and maximum among the 10 simulations with different sets of random parameters. 
These extremes obtained in the sequential MPS simulations and the pTEBD simulations are consistent with each other. 
In some cases [see Fig.~\ref{fig:Fs_vs_D}(b)], the maximum value obtained in the pTEBD simulation can be larger than 
that obtained in the sequential MPS simulation. To present more clearly the fidelities obtained using these two algorithms, 
we also show the results for a quantum circuit with the same single set of random parameters 
in Appendix~\ref{app:Fs_fixed_seed}.

For the unitary quantum dynamics described by a quantum circuit with local gates, the quantum entanglement 
propagates only within a finite range of space after applying each single layer of gates. Therefore, in the pTEBD algorithm, 
$\Lambda^{[i-1]}$ and $\Lambda^{[i+1]}$ are still expected to be a good approximation for the environments of each local 
two sites $i$ and $i+1$, although the canonical form deviates globally in general. 
This is probably the reason for the high precision of the pTEBD algorithm and the improvement with additional PtSU steps 
that suppress the deviation of the canonical form.

\begin{figure*}[!ht]
    \centering
    \includegraphics[width = \linewidth]{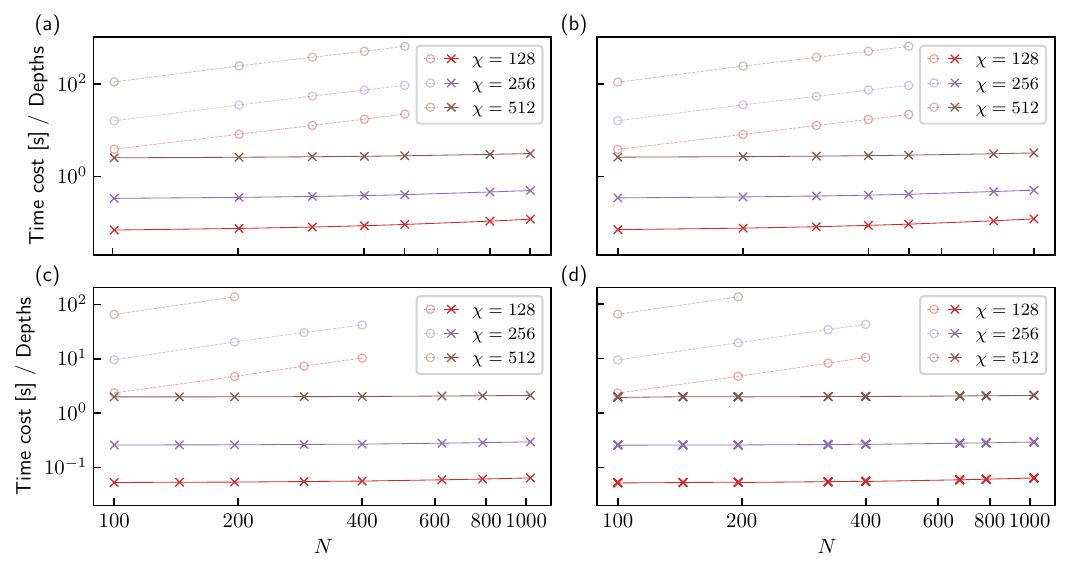}
    \caption{\label{fig:Ts_vs_N}
    Elapsed time per circuit layer
    (averaged over 10 simulations with different sets of random parameters) versus the system size 
    $N$ ($L_x=L_y$ in 2D cases) obtained using the sequential MPS algorithm (open circles) and the pTEBD algorithm (crosses) 
    with a fixed MPS bond dimension $\chi$:  
    (a) RQC-1D with 300 physical circuit layers, (b) PQC-1D with 300 physical circuit layers, 
    (c) RQC-2D with 100 physical circuit layers, and (d) PQC-2D with 100 physical circuit layers. 
    Note that the number of circuit layers in (c) and (d) is the total number of layers of the quantum circuit after recompiling 
    it to fit an MPS 1D path,  
    i.e., the compiled circuit depth (for more details, see Appendix~\ref{app:qc_setup}). 
    We set $g=0$ for the pTEBD simulations since the performance showing the almost perfect weak scaling does not depend on $g$. 
    } 
\end{figure*}

\subsubsection{\label{sec:ptebd_stab} Numerical stability}

After establishing the accuracy of the pTEBD algorithm, let us now demonstrate the crucial role of the wavefunction norm 
stabilization procedure introduced in Sec.~\ref{sec:thm2} in maintaining numerical stability during pTEBD simulations. 
To this end, we systematically evaluate the wavefunction norm in the simulations of 1D (2D) quantum circuits, containing 
up to 1000 (100) physical circuit layers,
with and without the wavefunction norm stabilization procedure. 
The results are summarized in Fig.~\ref{fig:MPSNorms_vs_D}.

In cases without the wavefunction norm stabilization procedure (see open squares in Fig.~\ref{fig:MPSNorms_vs_D}), 
the wavefunction norm monotonically and exponentially decays with increasing $D$. 
While it slightly increases with the bond dimension $\chi$, 
a tiny value, as small as $10^{-14}$, is reached after simulating a quantum circuit with around 200 and 20 physical circuit layers
both in 1D and 2D cases, respectively. This indicates that a renormalization procedure, involving sequential calculations and 
thus breaking the real-space parallelism, must be employed to prevent the exponentially decaying wavefunction norm 
and stabilize the simulation.

In sharp contrast, in cases with the wavefunction norm stabilization procedure (see open diamonds in 
Fig.~\ref{fig:MPSNorms_vs_D}), the wavefunction norm consistently remains close to one, resisting decay with increasing $D$, 
even when $\chi$ is small. This demonstrates that, in the simulation of unitary quantum dynamics, the numerical stability 
of the pTEBD algorithm can be maintained by the IPMC, particularly the wavefunction norm stabilization procedure, 
even though the MPS deviates from its canonical form. 
Hence, no sequential procedure is required throughout the entire pTEBD simulation, implying the achievement of scalability, i.e., 
perfect weak scaling, as will be examined in the next section.

\subsubsection{Performance and week scaling}

Finally, let us assess the performance of the pTEBD algorithm from two perspectives. On the one hand, we examine 
the simulation elapsed time per circuit layer, which is defined as the total simulation elapsed time divided by \textit{compiled} 
circuit depth (for more details, see Appendix~\ref{app:qc_setup}),
versus the system size $N$, evaluating its performance in weak scaling. 
We always distribute tensors of each 4 adjacent sites in an MPS on one computational node. 
When $N$ is not divisible by 4, the remaining tensors are distributed on the last node.
On the other hand, we compare the elapsed time required to achieve a given simulation precision, 
measured by wavefunction fidelity $\mathcal{F}$, in both sequential MPS and pTEBD simulations. 

As illustrated in Fig.~\ref{fig:Ts_vs_N}, for a fixed bond dimension $\chi$, the pTEBD algorithm consistently exhibits nearly 
constant simulation times for quantum circuits with various system sizes, achieving excellent weak scaling performance. 
This is attributed to the absence of any real-space sequential procedures and the lack of a need for global data communication 
in the pTEBD algorithm [also see Fig.~\ref{fig:ptebd_algo}(d)]. 
In contrast, the elapsed time per circuit layer
in the sequential MPS algorithm increases linearly 
with the system size due to its unavoidable sweep procedure.
We also expect good strong scaling since, compared with the previous approach in Ref.~\cite{urbanek2016parallel}, 
the IPMC does not increase the computational complexity.

\begin{figure*}[!ht]
    \centering
    \includegraphics[width=0.98\linewidth]{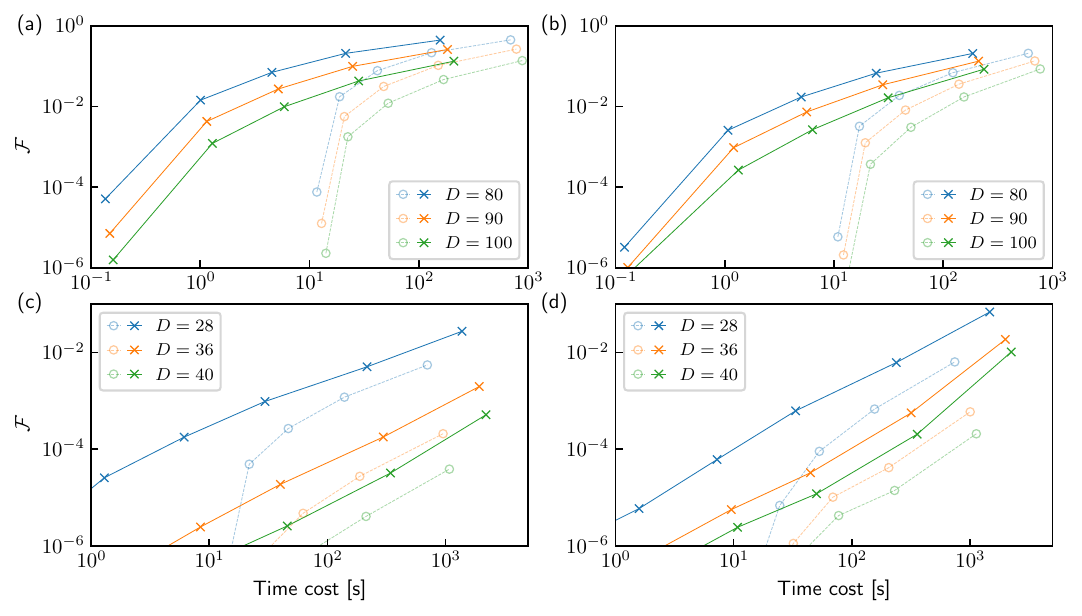}
    \caption{\label{fig:Fs_vs_T} 
    Wavefunction fidelity $\mathcal{F}$ versus elapsed time obtained using the sequential MPS algorithm (open circles) 
    and the pTEBD algorithm (crosses) with a fixed physical  circuit depth $D$ (for the definition, see Appendix~\ref{app:qc_setup}): 
    (a) RQC-1D with $N = 25$, (b) PQC-1D with $N = 24$, 
    (c) RQC-2D with $L_{x} = L_{y} = 5$ (i.e., $N = 25$), and (d) PQC-2D with $L_{x} = 4$ and $L_{y} = 6$ (i.e., $N = 24$). 
    These results are averaged over 10 simulations of the same quantum circuit but with different sets of random parameters. 
    We set $g=0$ for the pTEBD simulations.
    Note that the elapsed time increases simply because of the increase of bond dimensions.
    }
\end{figure*}

It might be more insightful to compare the performance between these two algorithms by studying their elapsed time to achieve 
the same simulation precision, a measure similarly adopted in a recent study of a 2D tensor network 
algorithm~\cite{adachi2020anisotropic}. This measure demonstrates the practical speedup of the pTEBD algorithm 
over its sequential counterpart. For this purpose, we reinterpret the fidelity results shown in Fig.~\ref{fig:Fs_vs_D} 
and replot these results 
as a function of the elapsed time of the simulation in Fig.~\ref{fig:Fs_vs_T}. Remarkably, in all cases, the pTEBD simulation reaches 
the same $\mathcal{F}$ in polynomially shorter time than the corresponding sequential 
MPS simulation (as also expected from Fig.~\ref{fig:Ts_vs_N}), indicating its efficient utilization of the parallel computing environment.

\section{\label{sec:sum}summary and discussion}

In summary, we have proposed an improved parallel MPS compression method that can accurately compress 
the dimensions of all the virtual bonds in a constant time, regardless of the system size. Simultaneously, it stabilizes 
the wavefunction norm of the MPS, converging to a value around 1 that is bounded from both sides. 
Although both the accuracy and the norm stabilization are mathematically 
proved under the assumption of the canonical form, we have demonstrated its feasibility in simulating unitary quantum dynamics, 
where the canonical form slightly deviates globally. Moreover, we have numerically shown that the deviated canonical form 
resulting from the parallel truncation procedure in the parallel MPS compression can be gradually restored 
through the subsequent PtSU steps. 
Additionally, we have provided a proof that 
at most $\frac{N}{2}$ PtSU steps can drive any MPS to its canonical form.

Utilizing the IPMC method, we have proposed a fully real-space parallelizable pTEBD algorithm for efficiently simulating 
unitary quantum dynamics. Furthermore, we systematically benchmarked the pTEBD algorithm by simulating typical 1D 
and 2D quantum circuits. Our results demonstrate that the pTEBD algorithm 
achieves nearly perfect 
weak scaling performance even for very deep quantum circuits with hundreds of circuit layers. 
Additionally, it achieves the same simulation precision in polynomially shorter time compared to 
the sequential MPS algorithm.

While existing MPS-based quantum computing simulation methods have been recognized as highly 
efficient for simulating NISQ devices, their applications to quantum circuits with hundreds of qubits and circuit layers remains 
challenging due to the inherent limitation of linear time complexity with the system size. 
In contrast, the pTEBD algorithm proposed in this study can harness the abundant computing resources offered by 
current supercomputing systems to address these demanding tasks, providing a practical way for exploring quantum computing 
on large NISQ devices. 

As we have demonstrated, the pTEBD algorithm exhibits the same simulation capability 
as other sequential MPS algorithms. However, 
due to the inherent limitation of MPS expressibility, 
simulating quantum circuits with higher-dimensional structures 
and/or rapidly accumulating large entanglement remains 
challenging~\cite{boixo2018characterizing,arute2019quantum,kim2023evidence}.
Therefore, it is of great interest to extend the present parallel approach to recently developed approximate contraction 
algorithms for more complex tensor networks~\cite{contracting2020pan,hyperoptimized2024gray} or to integrate 
the pTEBD algorithm with other levels of parallelization, such as those utilized in quantum chemistry 
simulations~\cite{shang2022large}, which are feasible on current supercomputers.

\begin{acknowledgments}

We acknowledge the valuable discussions with Hong-Hao Tu, Kazuma Nagao, Hidehiko Kohshiro, and Shigetoshi Sota. Numerical 
simulations were performed using the high-performance tensor computing library \textit{GraceQ/tensor}~\cite{gqten} 
on the Supercomputer Fugaku installed at RIKEN Center for Computational Science 
(Project ID No.~hp220217 and No.~hp230293). This work is
supported by Grant-in-Aid for Scientific Research (A) (No.~JP21H04446) and Grant-in-Aid for Scientific Research (C)
(No.~JP22K03479) from MEXT, Japan, as well as JST COI-NEXT
(Grant No.~JPMJPF2221). Additionally, this work receives support from 
the Program for Promoting Research on the Supercomputer 
Fugaku (No.~JPMXP1020230411) from MEXT, Japan, and
from the COE research grant
in computational science from Hyogo Prefecture and Kobe
City through the Foundation for Computational Science.
This work is also supported in part by the New Energy and Industrial Technology Development Organization (NEDO) 
(No.~JPNP20017). 

\end{acknowledgments}

\appendix

\section{\label{app:qc_setup}Details of the quantum circuits adopted in the benchmarking simulations}

To benchmark the pTEBD algorithm, we select random quantum circuits (RQCs) and parametrized 
quantum circuits (PQCs), the latter being commonly used in variational quantum algorithms, as the simulation tasks. 
RQCs have been widely employed in benchmarking NISQ 
devices~\cite{arute2019quantum,wu2021strong} and quantum computing simulations~\cite{liu2021closing}. 
PQCs, being more featured circuits, can reflect the practical performance of devices and 
simulations for more realistic 
computational tasks. Therefore, these benchmarks allow for a thorough evaluation of the pTEBD algorithm's performance 
in the most realistic and representative cases. In this Appendix, we provide explanations for the construction of these 
quantum circuits. More detailed features and applications of these quantum circuits can be found in the corresponding references 
cited below.

\subsection{Physical circuit depth vs. compiled circuit depth}

Physical circuit depth is defined as the total number of physical circuit layers. Typically, for a given quantum circuit featuring repeating 
structural units, as observed in the cases of RQCs and PQCs considered in this study, a physical circuit layer is determined by 
a single type of repeating unit. For instance, in the cases of RQCs, a single physical circuit layer is formed by applying single-qubit 
random gates to each qubit, followed by the application of two-qubit entangling gates with a specific covering pattern. 
The specific definition of the physical circuit layer in each quantum circuit adopted in this study will be provided in the following 
subsections.

To enable a classical simulation of a quantum circuit with an MPS representation using a TEBD-type algorithm, a physical circuit 
layer must be recompiled to one (or several) compiled circuit layer(s), wherein all two-qubit gates exclusively apply to 
neighboring qubits aligned on the MPS 1D path. In this study, for 1D quantum circuits, one physical circuit layer is directly mapped 
to one compiled circuit layer. However, for 2D quantum circuits, a portion of physical circuit layers (details provided below) 
must be mapped to several compiled circuit layers, thereby increasing the effective circuit layers in the simulation. 
The total number of compiled circuit layers is denoted as the compiled circuit depth.
As a single-qubit gate can always be contracted exactly into a nearby two-qubit gate during the recompiling procedure 
in the quantum circuit, we only consider two-qubit gates in defining and counting compiled circuit layers.

\subsection{One-dimensional random quantum circuit (RQC-1D)}

The circuit structure of RQC-1D is illustrated in Fig.~\ref{fig:rqc_pqc_1d}(a). Here, we adopt the same construction as 
in Ref.~\cite{zhou2020what}. The single-qubit random gate at qubit $r$ in the $l$th circuit layer, 
denoted by squares in the figure, is defined as 
\begin{align}
    U_{r}^{(l)} =& \exp[-i \theta_r^{(l)} (\sigma^{x}_r\sin \alpha_r^{(l)} \cos \phi_r^{(l)} \nonumber \\
    & + \sigma^{y}_r \sin \alpha_r^{(l)} \sin \phi_r^{(l)} + \sigma^{z}_r\cos\alpha_r^{(l)} )]~,
\end{align}
where $\sigma^{x}_r$, $\sigma^{y}_r$, and $\sigma^{z}_r$ are Pauli matrices, 
and $\alpha_r^{(l)}$, $\theta_r^{(l)}$, and $\phi_r^{(l)}$ are 
uniformly distributed real random numbers. 
After applying $U_{r}^{(l)}$ to all qubits in the $l$th circuit layer, the quantum entanglement is generated by subsequently 
applying a layer of controlled-Z (CZ) gates, completing one physical layer of the circuit construction. The CZ gates apply to 
$(0, 1), (2, 3),\cdots,(N-3, N-2)$ qubit pairs for odd physical layers and on $(1, 2), (3, 4),\cdots,(N-2, N-1)$ qubit pairs for even 
physical layers. 
The total number $N$ of qubits is restricted to odd, and the total number $D$ of physical circuit layers is restricted to even. Note that these RQCs (including the RQC-2D introduced below) are different from those used 
in Google's quantum supremacy experiment~\cite{arute2019quantum}, in which fSim gates are adopted to generate larger 
entanglement. More detailed discussion of these two kinds of RQCs can be found in Ref.~\cite{zhou2020what,pan2022solving}.

\begin{figure}[!t]
    \centering
    \includegraphics[width=\linewidth]{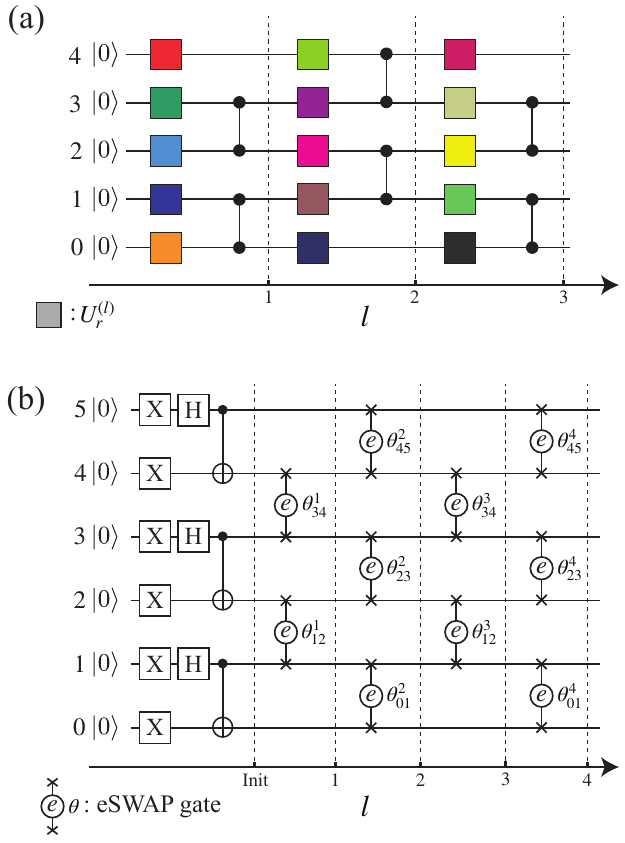}
    \caption{\label{fig:rqc_pqc_1d}
    Schematic figures for (a) RQC-1D and (b) PQC-1D. In (a), the colored squares represent single-qubit random gates. 
    In (b), the first layer of the quantum circuit is dedicated to initializing the state as a product of singlet dimers. 
    Dashed lines indicate each physical circuit layer. 
    }
\end{figure}

\subsection{One-dimensional parametrized quantum circuit (PQC-1D)}

For a more practical setting, we also consider the Hamiltonian variational \textit{Ansatz} (HVA)~\cite{wecker2015progress}, 
which has been widely studied in quantum computing for quantum many-body systems, and apply it to the nearest-neighbor 
spin $S=1/2$ Heisenberg model. Initializing the qubit register to form a product of singlet dimers on $(0, 1), (2, 3),\cdots,(N-2, N-1)$ 
qubit pairs [see Fig.~\ref{fig:rqc_pqc_1d}(b)], a layer of eSWAP gates~\cite{loss1998quantum,divincenzo2000universal,brunner2011two,lloyd2014quantum,lau2016universal,seki2020symmetry,sun2023efficient}, 
$U_{ij}^{(l)}(\theta_{ij}^{(l)}) = \exp (-i \theta_{ij}^{(l)} P_{ij}/2)$ with $P_{ij}$ being the SWAP gate acting at qubits $i$ and $j$, is 
applied to $(1, 2), (3, 4),\cdots,(N-3, N-2)$ [$(0, 1), (2, 3),\cdots,(N-2, N-1)$] qubit pairs for the $l$th physical layer with odd (even) 
$l$. In using the HVA for variational quantum algorithms~\cite{seki2020symmetry,sun2023efficient}, $\{\theta_{ij}^{(l)}\}$ is a set of 
variational parameters to be optimized. Here, we assign uniformly distributed real random numbers to these parameters 
$\{\theta_{ij}^{(l)}\}$ in our benchmark simulations. The resulting circuit structure of PQC-1D is illustrated in 
Fig.~\ref{fig:rqc_pqc_1d}(b). 
The total number $N$ of qubits and the total number $D$ of physical circuit layers are both restricted to even.

\subsection{Two-dimensional random quantum circuit (RQC-2D)}

RQC-2D is the direct extension of RQC-1D on a 2D square lattice with the number of qubits $N = L_x \times L_y$. 
The physical circuit layers are applied in an ABCDABCD~$\cdots$ pattern, as illustrated in Fig.~\ref{fig:qc_2d_abcd}, where each 
rectangle indicates the qubits on which a CZ gate acts. Additionally, single-qubit random gates 
$\{U_r^{(l)}(\theta_{ij}^{(l)})\}$ are applied to all qubits in 
each physical circuit layer, similar to the case of RQC-1D. The number of physical circuit layers is restricted to a multiple of 4. 
However, there are no specific restrictions on $L_{x}$ and $L_{y}$. In cases when $L_x$ ($L_y$) is odd, no two-qubit gates are 
applied on qubits at the right and left (top and bottom)
edges in physical circuit layers C and D (A and B), respectively,  
as shown in Fig.~\ref{fig:qc_2d_abcd} for $L_x=L_y=5$.

\begin{figure}[!ht]
    \centering
    \includegraphics[width=0.95\linewidth]{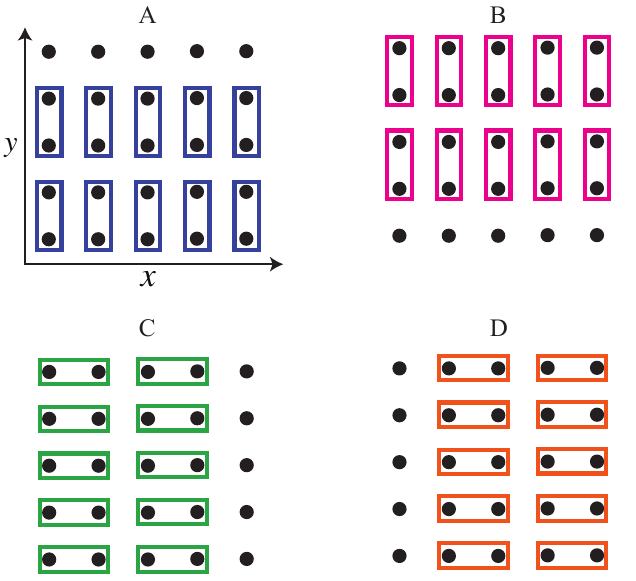}
    \caption{\label{fig:qc_2d_abcd}
    Locations of two-qubit gates, indicated by rectangles, in four distinct physical circuit layers A, B, C, and D 
    for both RQC-2D and PQC-2D. Solid dots represent qubits arranged in a 2D square lattice. 
    The two-qubit gates consist of CZ gates for RQC-2D and eSWAP gates for PQC-2D.      
    Additionally, single-qubit random gates are applied to all qubits in each physical circuit layer for RQC-2D, 
    similar to the case of RQC-1D shown in Fig.~\ref{fig:rqc_pqc_1d}(a). 
    For PQC-2D, the first physical layer A is dedicated to forming a product of singlet dimers, as in the case of PQC-1D 
    shown in Fig.~\ref{fig:rqc_pqc_1d}(b). 
    In this example, we set $L_{x} = 5$ and $L_{y} = 5$. 
    }
\end{figure}

In contrast to 1D quantum circuits, where qubits naturally align in a 1D path suitable for the MPS representation, 
a mapping of qubits for 2D quantum circuits to a 1D path suitable for the MPS representation must be determined. 
Here, we choose the path shown in Fig.~\ref{fig:qc_2d_recompile}(a) to ensure that two-qubit gates in physical circuit layers A and B 
act on neighboring qubits in the sense of the MPS 1D path. 
Under this mapping, the two-qubit gates in physical circuit layers C and D become long-distance gates, 
necessitating a recompilation of these two physical circuit layers for MPS-based simulations to ensure that gates apply only to 
neighboring qubits on the MPS 1D path. 
As an example, the recompiling procedure for the $L_{x} = L_{y} = 4$ case is illustrated in 
Fig.~\ref{fig:qc_2d_recompile}(b). 
Note that the circuit depth increases from 2 (physical circuit depth) to $3(L_{y} - 1) + 2$ (compiled circuit depth) after this recompiling. 
The physical circuit depths and compiled circuit depths for several 2D quantum circuits studied in 
Fig.~\ref{fig:Fs_vs_D}--Fig.~\ref{fig:Fs_vs_T} are provided in Table~\ref{table}. 
Although we focus on the square lattice, an extension to other lattices such as  
the Bristlecone lattice used in Ref.~\cite{arute2019quantum} is straightforward.

\begin{figure}[!ht]
    \centering
    \includegraphics[width=0.95\linewidth]{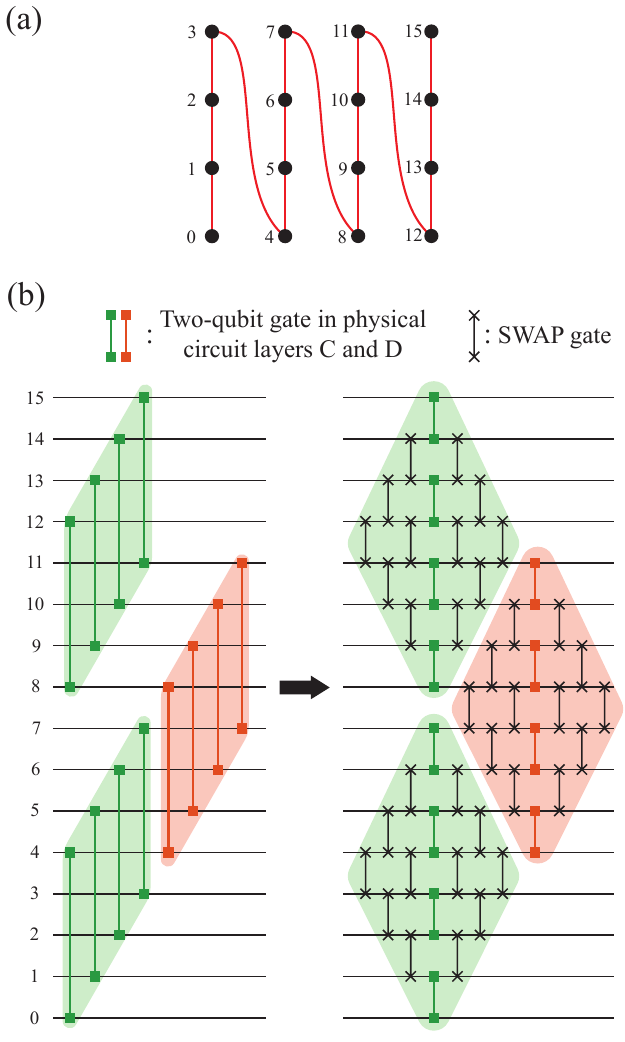}
    \caption{\label{fig:qc_2d_recompile}
    (a) Mapping of qubits for a 2D quantum circuit to a 1D path suitable for its MPS representation. 
    Solid dots represent qubits arranged in a 2D square lattice with $L_{x} = L_{y} = 4$. 
    The numbers beside solid dots indicate the indices of qubits. 
    (b) Arrangement of two-qubit gates in physical circuit layers C and D for a 2D quantum circuit with $L_{x} = L_{y} = 4$ 
    (left-hand side) and the equivalent recompiled quantum circuit containing only two-qubit gates applying to neighboring qubits 
    but with many additional SWAP gates. The number of each qubit register corresponds to that in (a). 
    The shades of the same color indicate the correspondence between gates involved before and after the recompiling 
    procedure for each physical circuit layer.
    Although we consider the $L_{x} = L_{y} = 4$ case here as an example, it is straightforward to extend to other cases. 
    }
\end{figure}

\begin{table}
  \caption{
    Physical and compiled circuit depths for several 2D quantum circuits  
    studied in Fig.~\ref{fig:Fs_vs_D}--Fig.~\ref{fig:Fs_vs_T}. Here, a 2D quantum circuit is arranged in a square 
    lattice with the number of qubits $N=L_x\times L_y$. 
    \label{table} }
  \begin{tabular}{ccccc}
    \hline
    \hline
    $N$  & $L_x$ & $L_y$ & physical circuit depth  & compiled circuit depth \\ \hline
    24 & 4  & 6 & 28 & 133 \\
    24 & 4  & 6 & 60 & 285 \\
    25 & 5  & 5 & 28 & 112 \\
    25 & 5  & 5 & 60 & 192 \\
    144 & 12  & 12 & 40 & 370 \\
    144 & 12  & 12 & 100 & 925 \\
    1024 & 32  & 32 & 100 & 2425 \\
    \hline
    \hline
  \end{tabular}
\end{table}

\begin{figure*}[!ht]
    \centering
    \includegraphics[width=\linewidth]{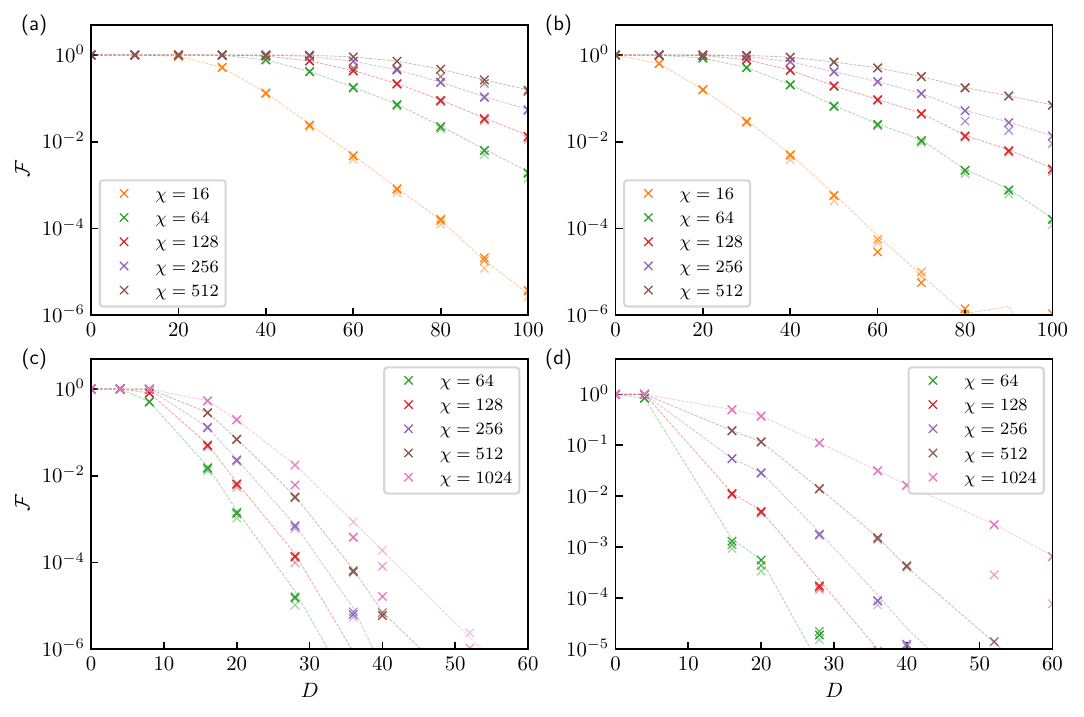}
    \caption{\label{fig:qcs_Fs_vs_D_48}
        Same as Fig.~\ref{fig:Fs_vs_D} but showing the results for a single set of random parameters, selected from 
        the 10 different sets of random parameters used for averaging in Fig.~\ref{fig:Fs_vs_D}. 
        }
\end{figure*}

\subsection{Two-dimensional parametrized quantum circuit (PQC-2D)}

Similar to RQC-2D, PQC-2D is the extension of PQC-1D on a 2D square lattice with the number of qubits $N = L_x \times L_y$. 
A slight difference from RQC-2D is found in the first physical circuit layer A, where we prepare a product of singlet dimers, similar 
to PQC-1D. Beyond this initial physical circuit layer A, the subsequent physical circuit layers incorporate eSWAP gates in a 
BCDABCDA~$\cdots$ pattern with parameters $\{\theta_{ij}^{(l)}\}$ (see Fig.~\ref{fig:qc_2d_abcd}). 
Similar to RQC-2D, we employ an MPS 1D path as depicted in Fig.~\ref{fig:qc_2d_recompile}(a). 
Consequently, two-qubits gates in physical circuit layers C and D must be recompiled to ensure that all gates apply only 
to neighboring qubits, as illustrated in Fig.~\ref{fig:qc_2d_recompile}(b).  
Due to the specific configuration, the value of $L_{y}$ in PQC-2D is restricted to an even number. 
As in the case of RQC-2D, the total number $D$ of physical circuit layers is a multiple of 4.

\section{\label{app:Fs_fixed_seed}Simulation accuracy of the pTEBD algorithm for a single instance of random parameters}

For a fair comparison of the wavefunction fidelity $\mathcal{F}$ obtained by the pTEBD algorithm 
and the sequential MPS algorithm, 
here we compare the results for a specific set of random parameters rather than averaging over 10 simulations 
with different sets as in Fig.~\ref{fig:Fs_vs_D}. 
This particular set is chosen among the 10 different sets used in Fig.~\ref{fig:Fs_vs_D}. 
The corresponding results are denoted as $\mathcal{F}_{\rm pTEBD}$ and  
$\mathcal{F}_{\rm MPS}$ obtained by the pTEBD and sequential MPS algorithms, respectively.

As shown in Fig.~\ref{fig:qcs_Fs_vs_D_48}, 
in most cases, we observe behavior consistent with that found in Fig.~\ref{fig:Fs_vs_D}: 
$\mathcal{F}_{\rm pTEBD} \approx \mathcal{F}_{\rm MPS}$, 
and $\mathcal{F}_{\rm pTEBD}$ increases with the number $g$ of PtSU steps. 
Additionally, we note that $\mathcal{F}_{\rm pTEBD}$ outperforms $\mathcal{F}_{\rm MPS}$ in some cases 
with small bond dimensions and larger circuit depths, while $\mathcal{F}_{\rm pTEBD}$ decreases with more appended PtSU 
steps in certain cases. In practice, we can adjust $g$ to achieve optimal performance in the pTEBD simulation.

\section{\label{app:QFT_pTEBD}Additional benchmark test of pTEBD}

Quantum Fourier transformation (QFT)~\cite{coppersmith1994approximate} is a fundamental algorithm in fault-tolerant quantum 
computing (FTQC) and is typically integrated into other FTQC algorithms, such as Shor's algorithm~\cite{shor1994algorithms} and 
the quantum phase estimation algorithm~\cite{kitaev1995quantum}. In this appendix, we simulate the QFT applied to a random state 
using both the sequential MPS algorithm and the pTEBD algorithm to assess the accuracy and scalability of the pTEBD algorithm 
in the simulation of more irregularly structured quantum circuits compared with those studied in Sec.~\ref{sec:ptebd_2}.

We initialize a random MPS with bond dimension $\chi_{0} = 10$ and apply a QFT circuit to this random state. The QFT circuit is 
generated and further compiled to a qubits layout with linear connectivity using \texttt{Qiskit}~\cite{qiskit}. For the compilation, the 
\texttt{transpile()} function is used with the single-qubit rotation gate $U(\theta, \phi, \lambda)$ and the controlled-NOT gate CX as 
the basis gates, and with the optimization level 3. The physical circuit depths and compiled circuit depths of the QFT circuits for 
different system sizes are summarized in Table~\ref{tab:qft_circ_depth}. We then employ the sequential MPS algorithm and 
the pTEBD algorithm to simulate these compiled circuits.

\begin{table}
  \caption{
    Physical and compiled circuit depths for QFT circuits with various numbers ($N$) of qubits. \label{tab:qft_circ_depth} }
  \begin{tabular}{ccc}
    \hline
    \hline
    $N$ & physical circuit depth  & compiled circuit depth \\ \hline
    16 & 136 & 946 \\
    20 & 210 & 1482 \\
    24 & 300 & 2219 \\
    28 & 406 & 3034 \\
    32 & 528 & 3882 \\
    48 & 1176 & 8430 \\
    64 & 2080 & 14670 \\
    80 & 3240 & 22514 \\
    96 & 4656 & 31493 \\
    \hline
    \hline
  \end{tabular}
\end{table}

First, we examine the simulation accuracy. For this purpose, we fix the number of qubits $N = 16$ with the corresponding compiled 
circuit depth being 946. The results of the simulation accuracy, quantified by the wavefunction fidelity $\mathcal{F}$ 
in Eq.~(\ref{eq:F}), are shown in Fig.~\ref{fig:QFT_fides}. We observe that the wavefunction fidelity in the pTEBD 
simulation always approximates very accurately the wavefunction fidelity obtained from the sequential MPS simulation, 
even in the cases with relatively small bond dimensions. This indicates that the accuracy of the pTEBD algorithm is comparable to 
that of the sequential MPS algorithm, as in the cases for more regularly structured quantum circuits 
before recompilation, studied in Fig.~\ref{fig:Fs_vs_D}.

\begin{figure}[!ht]
  \centering
  \includegraphics[width=\linewidth]{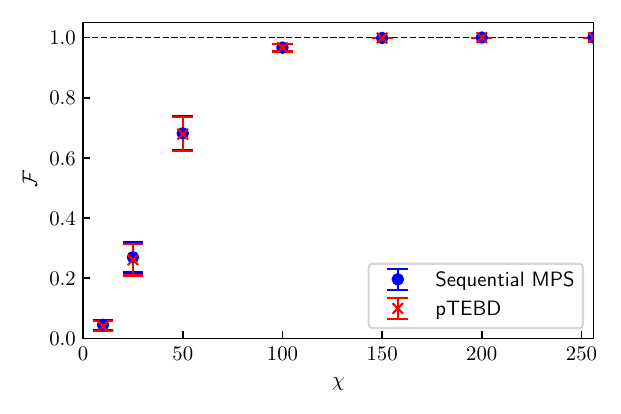}
  \caption{\label{fig:QFT_fides}
  Wavefunction fidelity $\mathcal{F}$ as a function of MPS bond dimension $\chi$ in simulations of the QFT circuit with 
  $N=16$ using the sequential MPS algorithm and the pTEBD algorithm. We set $g=0$ for the pTEBD simulations. 
  The results for each $\chi$ are obtained from the average over 10 different random initial states.
  }
\end{figure}

Furthermore, we compare the performance of these two algorithms in simulating QFT circuits. Figure~\ref{fig:QFT_elapses} shows 
the elapsed time per compiled circuit layer against the number $N$ of qubits. The elapsed time per layer increases linearly with 
$N$ in the sequential MPS simulations, while it remains approximately constant, especially for larger $N$, in the pTEBD simulations. 
The very slowly increase in elapsed time per layer in the pTEBD simulations might be due to the increase 
in inter-node communications, as we always distribute 16 qubits to each node in these simulations. This performance analysis 
demonstrates that the pTEBD algorithm can also achieve good weak scaling even when simulating more general, 
unstructured circuits, 
consistent with the results for RQC-2D and PQC-2D in Figs.~\ref{fig:Ts_vs_N}(c) and \ref{fig:Ts_vs_N}(d), where the quantum circuits 
after recompilation are not perfectly regularly structured. 

\begin{figure}[!ht]
  \centering
  \includegraphics[width=\linewidth]{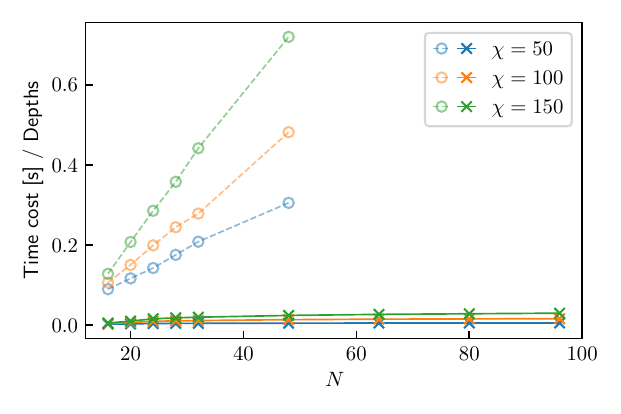}
  \caption{\label{fig:QFT_elapses}
  Elapsed time per circuit layer versus the number $N$ of qubits obtained using the sequential MPS algorithm (open circles) 
  and the pTEBD algorithm (crosses), with a fixed MPS bond dimension $\chi$ in the simulations of QFT circuits. 
  For the pTEBD simulations, we set $g=0$ and distribute 16 qubits to each node. }
\end{figure}

\end{document}